\documentclass[transaction]{IEEEtran}
\IEEEoverridecommandlockouts
\usepackage{hhline}
\usepackage{amsmath}
\usepackage{amsthm}
\usepackage[utf8]{inputenc}
\usepackage[english]{babel}
\usepackage{amsfonts}
\usepackage{graphicx}
\usepackage{epsfig}
\usepackage[font=footnotesize]{caption}
\usepackage{subfig}
\usepackage{float}
\usepackage{epstopdf}
\usepackage{array}
\usepackage{multirow}
\usepackage{setspace}
\usepackage{color}
\usepackage[section]{placeins}
\usepackage{url}
\usepackage{algorithm}
\usepackage{algpseudocode}
\usepackage{multicol}
\usepackage{lipsum}
\usepackage{soul}
\usepackage{comment}
\usepackage{balance}
\usepackage{color}
\usepackage[normalem]{ulem}
\usepackage{amssymb}
\usepackage{array}
\usepackage{soul}

\usepackage{xcolor}
\usepackage{cite}
\usepackage{tikzpagenodes}

\newcommand{\etal}{\textit{et al.}}
\newtheorem{theorem}{Theorem} 
\newtheorem{corollary}{Corollary}
\begin{document}

\title{Wake-Up Radio based Access in 5G under Delay Constraints: Modeling and Optimization}
\author{Soheil~Rostami, Sandra Lagen, M\'ario Costa,~\IEEEmembership{Member,~IEEE}, \\ Mikko Valkama,~\IEEEmembership{Senior~Member,~IEEE}, and Paolo Dini

\thanks{S. Rostami is with Huawei Technologies Oy (Finland) Co. Ltd, Helsinki, Finland, and also with Tampere University, Finland. E-mail: soheil.rostami1@huawei.com, soheil.rostami@tuni.fi.}
\thanks{S. Lagen  and P. Dini are with Centre Tecnol\`ogic de Telecomunicacions de Catalunya (CTTC/CERCA), Barcelona, Spain. E-mails: \{sandra.lagen, paolo.dini\}@cttc.es.}
\thanks{M. Costa is with Huawei Technologies Oy (Finland) Co. Ltd, Helsinki, Finland. E-mail: mariocosta@huawei.com.}
\thanks{M. Valkama is with the Department of Electrical Engineering, Tampere University, Finland. E-mail: mikko.valkama@tuni.fi.}
\thanks{{Limited subset of early stage results presented at IEEE GLOBECOM, Dec. 2019 \cite{Globecomnew}.}}
\vspace{-5mm}}

\maketitle
\begin{tikzpicture}[remember picture,overlay]
\footnotesize
\node[align=center,color=red] at ([yshift=3em]current page text area.north) {This is the author's version of an article that has been accepted for publication in IEEE Transactions on Communications.};
\node[align=center,color=red] at ([yshift=2em]current page text area.north) {Changes were made to this version by the publisher prior to publication.};
\node[align=center,color=black] at ([yshift=-2em]current page text area.south) {Copyright (c) 2019 IEEE. Personal use of this material is permitted. However, permission to use this material for any other purposes must be \\ obtained from the IEEE by sending a request to pubs-permissions@ieee.org.};
\end{tikzpicture}

\begin{abstract}
  Recently, the concept of wake-up radio based access has been considered as an effective power saving mechanism for 5G mobile devices. In  this  article,  the average power consumption of a wake-up radio enabled mobile device is analyzed and modeled by using a semi-Markov process. Building on this, a delay-constrained optimization problem is then formulated, to maximize the device energy-efficiency under given latency requirements, allowing the optimal parameters of the wake-up scheme to be obtained in closed form.  The provided numerical results show that, for a given delay requirement, the proposed solution is able to reduce the power consumption by up to {40\%} compared with an optimized discontinuous reception (DRX) based reference scheme.
\end{abstract}
\vspace{-2mm}

\IEEEpeerreviewmaketitle

\section{Introduction}
\label{sec:intro}
In order for the emerging fifth generation (5G) mobile networks to satisfy the ever-growing needs for higher data-rates and network capacities, while simultaneously facilitating other quality of service (QoS) improvements, computationally-intensive physical layer techniques and high bandwidth communication are essential~\cite{nr2017}, \cite{Boccardi2014}. At the same time, however, the device power consumption tends to increase which, in turn, can deplete the mobile device's battery very quickly. {Moreover, it is estimated that feature phones and smartphones consume  2 kWh/year and 7 kWh/year, respectively,    based on charging every 60$^{th}$ hour equal to 40\% of battery capacity every day and a standby scenario of 50\%  of the remaining time \cite{Fehske}. Also, the   carbon footprints of production of feature phones  and smartphones are  estimated to be    18 kg and 30 kg CO2e per  device, respectively, which still is the major contributor  of CO2 emission of mobile communication systems   \cite{Fehske}. }

In general, battery lifetime is one of the main issues that mobile device consumers consider important from device usability point of view~\cite{Qualcomm2013}. However, since the evolution of battery technologies tends to be slow~\cite{Lauridsenthesis}, the energy efficiency of the mobile device's main functionalities, such as the cellular subsystem, needs to be improved~\cite{Qualcomm2013,Carroll2010,Lauridsen2015}. Furthermore, since the data traffic has been largely downlink-dominated~\cite{ITU2015},  the  power saving mechanisms for cellular subsystems in receive mode are of great importance. 

The 3rd generation partnership project (3GPP) has specified  discontinuous reception (DRX) as one of the \emph{de facto} energy saving mechanism for long-term evolution (LTE), LTE-Advanced and 5G New Radio (NR) networks \cite{TS38300,  erik, tr36.213, TS38.30}. DRX allows the mobile device to reduce its energy consumption by switching off some radio  modules for long periods of time, activating them only for short intervals. To this end, the modeling and optimization of DRX mechanisms have  attracted a large amount of research interest in recent years. The authors in \cite{Liu} proposed an adaptive approach to configure DRX parameters according to  users’ activities, aiming to balance power saving and packet delivery latency. Koc \etal~formulated the  DRX mechanism as a  multi-objective optimization problem in \cite{koc}, satisfying the latency requirements of active traffic flows and the corresponding preferences for power saving.  In \cite{Mihov}, DRX is modeled as a semi-Markov process with three states (active, light-sleep, and deep-sleep), and the average power consumption as well as the average delay are calculated and optimized.  Additionally, the authors in \cite{Ramazanali} utilized exhaustive search over a large parameter set to configure all DRX parameters. Such a method may not be attractive from a computational complexity perspective for real-time/practical applications, however, it provides the  optimal DRX-based power consumption and is thus used in this article as a benchmark. 

To improve the device's energy-efficiency beyond the capabilities of ordinary DRX, the concept of wake-up radio based access has been discussed, e.g., in \cite{Lauridsenthesis, Demirkol2009,  Mazloum2014}. Specifically, in the cellular communications context, the wake-up radio based approaches have been recently discussed and described, e.g., in~\cite{Globecom} and \cite{Lauridsen2016}. 
 In such concept, the mobile device monitors only  a narrowband control channel signaling referred to as wake-up signaling at specific time instants and subcarriers, in an OFDMA-based radio access systems such as LTE or NR, in order to decide whether to process the actual upcoming physical downlink control channel (PDCCH) or discard it. Compared to DRX-based systems, this directly reduces the buffering requirements and processing of empty subframes as well as the corresponding power consumption. Furthermore, in~\cite{Globecom}, the concept of a low-complexity wake-up receiver (WRx) was developed to decode the corresponding wake-up signaling, and  to acquire the necessary time and frequency synchronization.  {A wake-up scheme that enhances the power consumption of machine type communications  (MTC) is introduced in 3GPP LTE Release-15 \cite{TS36.300}, which is based on a narrowband signal, transmitted over  the available symbols of configured subframes. It is also considered as the starting point of NR power saving study item in 3GPP NR  Release-17~\cite{NR_PS}.}

In general, the existing wake-up concepts and algorithms, such as those described in \cite{Globecom,Lauridsenthesis,Demirkol2009,Mazloum2014, Lauridsen2016,wpwrx, Tang, Wilhelmsson, Kouzayha,Oller,Aoudia, ponna2018saving, eee1, eee5}, build on static operational parameters that are determined by the radio access network at the start of the user’s session, and kept invariant, even if traffic patterns change. Accordingly, methods to optimize such parameters that characterize the employed wake-up scheme are needed, to further reduce energy consumption according to the traffic conditions. This is one of the main objectives of this paper. Specifically, the main contributions of this paper are as follows. Firstly, the wake-up radio based access scheme is modeled by means of a semi-Markov process. {In the model, we consider realistic WRx operation by introducing start-up and power-down periods of the baseband unit (BBU), false alarm and misdetection probabilities of the wake-up signaling, as well as the packet service time.} With such a model, the average power consumption and buffering delay can be accurately quantified and estimated for a given set of wake-up related parameters.  Secondly, by utilizing such a mathematical model, the  minimization   of terminal's power consumption under Poisson traffic model is addressed for a given delay constraint. As a result, a closed-form optimal solution  for   the operational parameters is obtained.  Furthermore, the range of packet arrival rates, for which the wake-up scheme is suitable and energy-efficient, is determined. Finally, simulation-based numerical results are provided in order to validate the proposed model and methods as well as to investigate the power consumption  of our proposed solution compared to the optimized DRX-based reference mechanism proposed in \cite{Ramazanali}. The approach described in  \cite{Ramazanali} is selected as the benchmark since it provides the optimal power consumption in DRX-based reference systems. Furthermore,  to the best of our knowledge, virtually all of the DRX-based literature ignores the start-up and power-down energy consumption, {and \cite{Ramazanali} also neglects the packet service time}. Therefore, we have modified the approach in \cite{Ramazanali} slightly in order to consider such additional energy consumption in the optimization. 

The rest of this paper is organized as follows.   Section~\ref{sec:Background} presents a brief review of the considered wake-up scheme and its corresponding parameters, and defines basic system assumptions.   In 
Section~\ref{sec:sysmodel}, we model the wake-up based access scheme by means of a semi-Markov process and derive the power consumption as well as the buffering delay. Then, building on these mathematical models, in Section~\ref{sec:prob}, the optimization problem is formulated and the optimal solution for minimum power consumption is found in closed-form.  These are followed by numerical results, remarks, and conclusions  in Sections~\ref{sec:eval}, \ref{sec:remark}, and~\ref{sec:conc}, respectively. Some details on the analysis related to the modeling of power consumption are reported in the Appendix. For readers' convenience, the most relevant variables used throughout this paper are listed in Table~\ref {tab:vari}. Terminology-wise, we use gNB to refer to the base-station unit and UE to denote the mobile device, according to 3GPP NR specifications \cite{TS38300}.

\begin{table}[!t] 
\scriptsize
\centering
\renewcommand{\arraystretch}{1.3}
\caption{Most important variables used throughout the article } 
\label{tab:vari}
\begin{tabular}{cclll}
\cline{1-2}
\multicolumn{1}{|c|} {Variable} & \multicolumn{1}{c|}{Definition }                                                              &  &  &  \\ \hhline{==~} 
\multicolumn{1}{|c|}{ $\text{PW}_{\text{1}}$ } & \multicolumn{1}{c|}{{power consumption of cellular module while WRx is active}}                                                               &  &  &  \\ \cline{1-2}
\multicolumn{1}{|c|}{ $\text{PW}_{\text{2}}$ } & \multicolumn{1}{c|}{\begin{tabular}[c]{@{}c@{}}power consumption of cellular module \\ at  active-decoding state (BBU is active) \end{tabular}} &  &  &  \\ \cline{1-2}
\multicolumn{1}{|c|}{ $\text{PW}_{\text{3}}$ } & \multicolumn{1}{c|}{\begin{tabular}[c]{@{}c@{}}power consumption of cellular module at\\   active-inactivity timer state (BBU is active) \end{tabular}}&  &  &  \\ \cline{1-2}
\multicolumn{1}{|c|}{ $\text{PW}_{\text{4}}$ } & \multicolumn{1}{c|}{{power consumption of cellular module at   sleep state}}                                                               &  &  &  \\ \cline{1-2}
\multicolumn{1}{|c|}{ $\text{S}_{k}$ } & \multicolumn{1}{c|}{UE state in the state machine modeling {($\text{S}_1$, \dots, $\text{S}_4$)}}                                                               &  &  &  \\ \cline{1-2}
\multicolumn{1}{|c|}{ $S(\tau_n)$ } & \multicolumn{1}{c|}{{UE state at the $\tau_n$ jump time}   }                                                          &  &  &  \\ \cline{1-2}
\multicolumn{1}{|c|}{ $\text{P}_{kl}$ } & \multicolumn{1}{c|}{transition probability from state $\text{S}_{k}$ to state $\text{S}_{l}$ }                            \\ \cline{1-2}
\multicolumn{1}{|c|}{ $\omega_{k}$ } & \multicolumn{1}{c|}{holding time for state  $\text{S}_{k}$   }                                     &  &  &  \\ \cline{1-2}
\multicolumn{1}{|c|}{$t_{p}$} & \multicolumn{1}{c|}{inter-packet arrival time }                                                               &  &  &  \\
\cline{1-2}
\multicolumn{1}{|c|}{$\lambda$} & \multicolumn{1}{c|}{packet arrival rate }                                                               &  &  &  \\ \cline{1-2}
\multicolumn{1}{|c|}{$t_{su}$} & \multicolumn{1}{c|}{start-up time of cellular module }                                                               &  &  &  \\ \cline{1-2}
\multicolumn{1}{|c|}{  $t_{pd}$ } & \multicolumn{1}{c|}{power-down time of cellular module}         &  &  &  \\ \cline{1-2}
\multicolumn{1}{|c|}{  $t_{w}$ } & \multicolumn{1}{c|}{wake-up cycle}                          &  &  &  \\ \cline{1-2}
\multicolumn{1}{|c|}{  $t_{i}$ } & \multicolumn{1}{c|}{length of inactivity timer}                          &  &  &  \\ \cline{1-2}
\multicolumn{1}{|c|}{  $t_{s}$ } & \multicolumn{1}{c|}{{ 
packet service time}}                          &  &  &  \\ \cline{1-2}
\multicolumn{1}{|c|}{  $t_{on}$ } & \multicolumn{1}{c|}{on-duration time}                          &  &  &  \\ \cline{1-2}
\multicolumn{1}{|c|}{  $\overline{\text{P}}_c$ } & \multicolumn{1}{c|}{average power consumption}                          &  &  &  \\ \cline{1-2}
\multicolumn{1}{|c|}{  $\overline{\text{P}}_b$ } & \multicolumn{1}{c|}{average power consumption over boundary constraint}                          &  &  &  \\ \cline{1-2}
\multicolumn{1}{|c|}{  $\overline{\text{D}}$ } & \multicolumn{1}{c|}{average buffering delay}                          &  &  &  \\ \cline{1-2}
\multicolumn{1}{|c|}{  $\overline{\text{D}}_{\max}$ } & \multicolumn{1}{c|}{maximum tolerable delay or delay bound}                          &  &  &  \\ \cline{1-2}
\multicolumn{1}{|c|}{$t_{w_b}$} & \multicolumn{1}{c|}{minimum feasible wake-up cycle over boundary constraint}                                                               &  &  &  \\ \cline{1-2}
\multicolumn{1}{|c|}{$t_w^*$} & \multicolumn{1}{c|}{optimal value of wake-up cycle}                                                               &  &  &  \\ \cline{1-2}
\multicolumn{1}{|c|}{$t_i^*$} & \multicolumn{1}{c|}{optimal value of inactivity timer}                                                               &  &  &  \\ \cline{1-2}
\multicolumn{1}{|c|}{$\lambda_t$} & \multicolumn{1}{c|}{turnoff packet arrival rate}                                                               &  &  &  \\ \cline{1-2}
\multicolumn{1}{|c|}{$\eta$} & \multicolumn{1}{c|}{relative power saving factor}                                                               &  &  &  \\ \cline{1-2}
\multicolumn{1}{|c|}{$\phi$} & \multicolumn{1}{c|}{{power consumption ratio of UE at S$_2$ and S$_3$}}                                                               &  &  &  \\ \cline{1-2}
\multicolumn{1}{l}{}    & \multicolumn{1}{l}{}             &  &  & 
 
\end{tabular}
\end{table} 

\section{Basic Wake-Up Radio Concept and Assumptions}
\label{sec:Background}

In the considered wake-up radio based scheme, or wake-up scheme (WuS) for short,  as presented   in \cite{Globecom}, the   mobile device is configured to monitor a narrowband wake-up signaling channel in order to enhance its battery lifetime.   Specifically, in every wake-up cycle (denoted by $t_w$), the WRx monitors the so-called physical downlink wake-up channel (PDWCH) for a specific on-duration time ($t_{on}$) in order to determine whether data has been scheduled or not. Occasionally, based on the interrupt signal from WRx, the BBU switches on, decodes both    PDCCH   and  physical downlink shared channel (PDSCH),  and  performs normal connected-mode procedures. The WuS can be adopted for both connected and idle  states of radio resource control \cite{Globecom}, and can be configured based on maximum tolerable paging delay that idle users may experience,  or alternatively, based on the delay requirements of a specific traffic type at connected state.

The wake-up signaling  per each WRx contains a single-bit control information, referred to as the wake-up indicator (WI), where a WI of $1$ indicates the WRx to wake up the BBU, because there is one or multiple packets to be received, while a WI of $0$ signals the opposite. Each WI  is code multiplexed with a user-specific signature to selected time-frequency   resources, as described in \cite{Globecom}.  When a WI of $1$ is sent to WRx, the network expects the target mobile device to decode the PDCCH with a time offset identical to that of start-up time  ($t_{su}$).   Fig. \ref{fig:drxvswrx} (a) and (b) depict the basic operation and representative power consumption behavior of the conventional DRX-enabled cellular module and that of the cellular module with WRx, respectively, at a conceptual level. As illustrated, the WuS eliminates  the unnecessarily wasted energy in the first and second DRX cycles, while also reducing the buffering delay compared to DRX. {Due to the specifically-designed narrowband signal structure of WuS, the WRx power consumption ($\text{PW}_{\text{1}}$) is much lower than that of BBU active, either due to packet decoding  ($\text{PW}_{\text{2}}$) or when inactivity timer is running ($\text{PW}_{\text{3}}$) \cite{Globecom}. The lowest power consumption is obtained in sleep state ($\text{PW}_{\text{4}}$). We consider that during the BBU active states, the power consumption due to packet decoding is larger than or equal to that of running inactivity timer (i.e., $\text{PW}_{\text{2}}\geq \text{PW}_{\text{3}})$, but both are fixed. For presentation purposes, we denote the ratio of such power consumption at BBU active states     as $\phi=\frac{\text{PW}_2}{\text{PW}_3}$, where $\phi\geq 1$.}

In general, because NR supports wide bandwidth operation, packets can be served in a very short time duration. In addition, in case the user packet sizes are small, packet concatenation in NR is used, so that all buffered packets in a relatively short wake-up cycle can be served in a single transmission time interval (TTI). Accordingly, we assume that radio-link control entity (located at the gNB)  concatenates all those packets arriving during the sleep state, and  as soon as the BBU is triggered on, the device (UE) can receive and decode the  concatenated packets for a service time of $t_s$, which equals to a single TTI. During the serving time, if there was a new packet arrival, the BBU starts serving the corresponding packet by the end of $t_s$. In case that there was no packet arrival by the end of $t_s$, the UE initiates its inactivity timer with a duration of $t_i$. After the inactivity timer is initiated, and if a new PDCCH message is received before the time expiration, the BBU enters the active-decoding state and serves the packet. However, if there is no PDCCH message received  before the expiration of the inactivity timer, a sleep period starts, the WRx-enabled cellular module switches to sleep state, and WRx operates according to its {wake-up cycle} \cite{Globecom}. For reference, in case of DRX, the BBU sleeps according to its short and long DRX patterns \cite{pm2,koc}.

The introduction of a PDWCH has two fundamental consequences, namely misdetections and false alarms \cite{Globecom}. In the  latter case, WRx wakes up in a predefined time instant,  and erroneously decodes a WI of $0$ as $1$, leading to  unnecessary BBU power consumption. The former, in turn, corresponds to the case where a WI of $1$  is sent, but WRx  decodes it incorrectly as $0$. Such    misdetection   adds an extra delay and wastes radio resources. {We denote the probability of  misdetection and the probability of false alarm as $P_{md}$ and $P_{fa}$, respectively.} The requirements for the probability of  misdetection  of PDWCH are eventually stricter than those of the probability of false alarm \cite{Globecom}. 

\begin{figure}[!t]
\centering
\includegraphics[scale=1.1]{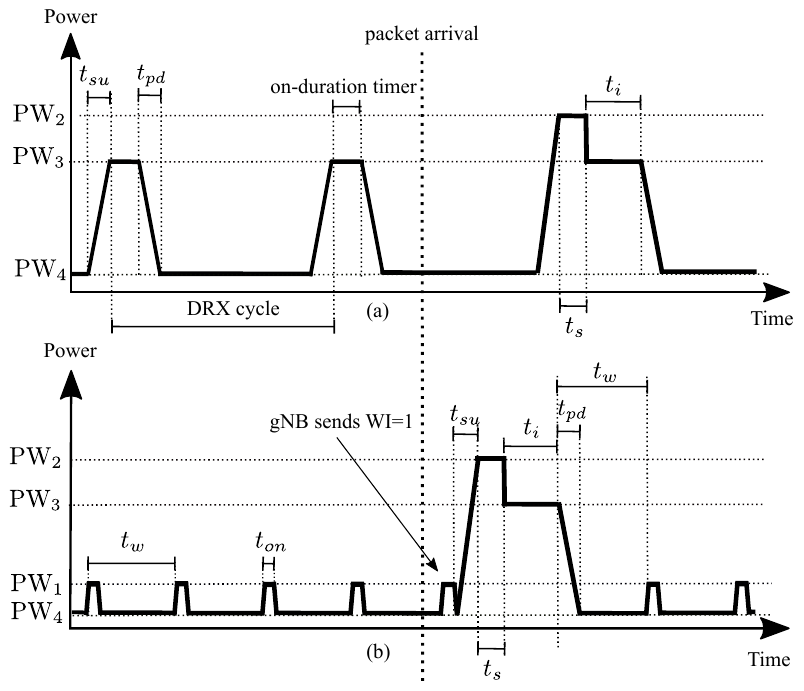}
\caption{{Power consumption profiles of (a) a typical DRX mechanism, and (b) a WRx-enabled cellular module.}  }
\label{fig:drxvswrx}
\end{figure}

One of the new features of 5G NR networks to reach their aggressive  requirements  is   latency-optimized frame structure with flexible numerology, providing  subcarrier spacings ranging from $15$ kHz up to $240$ kHz with a proportional change in cyclic prefix, symbol length, and slot duration  \cite{TS38300}.  Regardless of the numerology used, the length of one radio frame is fixed to $10$ ms and the length of one subframe is fixed to $1$ ms~\cite{nr2017}, as in 4G LTE/LTE-Advanced. However, in NR, the number of slots per subframe varies according to the numerology that is configured. Additionally, in order to support further reduced latencies, the concept of   mini-slot transmission is introduced in NR, and hence the TTI varies depending on the service type ranging from one symbol, to one slot, and to multiple slots~\cite{nr2017}.  In this work, in order to provide consistent and exact timing definitions, different time intervals of the wake-up related procedures  are defined as integer multiples of a TTI.  {Additionally, according to 3GPP,  the packet service time ($t_s$)   is one TTI, in which multiple packets can be concatenated.}  Furthermore, for the sake of clarity, a TTI duration of $1$ ms is taken   as the baseline system assumption for the WuS configurations, which then facilitates applying the proposed concepts also in future evolution of LTE-based systems.  

{Finally, it is important to note that from a system-level point of view, the configurable parameters of the WuS are the wake-up cycle ($t_w$) and the inactivity timer ($t_i$), whose values we will optimize in Section~\ref{sec:prob}. The remaining parameters ($t_{on}$, $t_{pd}$, $t_{su}$, $t_{s}$) depend on physical constraints and signal design, and accordingly we assume them to be fixed, i.e., the optimization will be done for fixed (given) values of $t_{on}$, $t_{pd}$, $t_{su}$ and $t_{s}$.}
\section{State Machine based Wake-up System Model}
\label{sec:sysmodel}
For mathematical convenience, the performance of the wake-up based system is studied and analyzed in the context of a  Poisson arrival process with a packet arrival rate of $\lambda$ packets per TTI. In the Poisson traffic model, each packet service session consists of a sequence of packets with exponentially distributed  inter-packet arrival  time ($t_p$) \cite{pm2}. 

{The power states of WuS are modeled as a semi-Markov process with four different states that correspond to WRx-ON (state $\text{S}_1$), active-decoding (state $\text{S}_2$), active-inactivity timer (state $\text{S}_3$),  and sleep (state $\text{S}_4$), as shown in Fig. \ref{fig:markov}.  At S$_1$, WRx monitors PDWCH, and if WRx receives WI=$0$, UE transfers  to    S$_4$, otherwise  (WI=$1$)  it transfers to  S$_2$.    At S$_2$, UE   decodes the packets for a fixed duration of $t_s$; if the device is scheduled before the end of $t_s$, it  starts decoding the   new packet, and remains at S$_2$, otherwise the device transfers to S$_3$.  At S$_3$, $t_i$ is running, and if the device is scheduled before the expiry of $t_i$, it enters in  S$_2$, otherwise the device transfers to S$_4$.  At S$_4$, the device   is in sleep state, and cannot receive any signal, as opposed to being fully-functional at S$_2$ and S$_3$.  Moreover,  at the end of a wake-up cycle in sleep period, the UE moves  to S$_1$.} {As noted already in the previous Section~\ref{sec:Background}, each state is associated with a different power consumption level, PW$_k$,   $k \in\{1,2,3, 4\}$.} 
  
\begin{figure}[t!]
\centering
\includegraphics[scale=1.6]{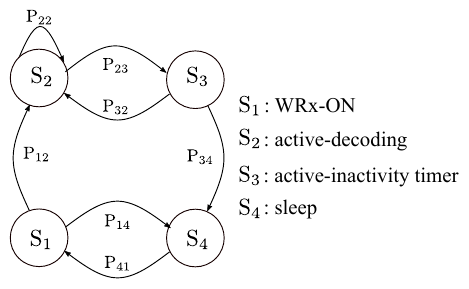} 
\caption{{Semi-Markov process for the state transitions of the wake-up scheme, with states $\text{S}_1$, $\text{S}_2$, $\text{S}_3$ and $\text{S}_4$.}}
\label{fig:markov}
\end{figure}

\textbf{Transition probabilities}: The transition probability from UE state   ${\text{S}_{k}}$ to ${\text{S}_{l}}$  (${\text{P}_{kl}}$) is defined as 
\begin{equation}
{\text{P}_{kl}}=\lim_{n\to\infty}  \Pr(S(\tau_n)={\text{S}_{l}} | S(\tau_{n-1})={\text{S}_{k}}),
\end{equation}
\noindent where $S(\tau_n)$ is the UE state at the $\tau_n$ jump time\footnote{{In a semi-Markov process, $S(t)$ is a stochastic process with a finite set of states (${\text{S}_1}, \dots, \text{S}_4$ in our case), having step-wise trajectories with jumps at times $0<\tau_1<\tau_2...$, and its values at the jump times ($S(\tau_n)$) form a Markov chain.
}}.

{When  the UE is at  ${\text{S}_{1}}$, it moves to ${\text{S}_{2}}$ either because of    false alarm or correct detection; otherwise, it moves to ${\text{S}_{4}}$. Accordingly,  ${\text{P}_{12}}$ and ${\text{P}_{14}}$ can be expressed as
  \begin{equation}
  \begin{split}
      {\text{P}_{12}}=\Pr[t_{p}>t_{w}|\text{S}_1]P_{fa}+
\Pr[t_{p}\leq t_{w}|\text{S}_1](1-P_{md})  \\
 =e^{-\lambda t_w}P_{fa}+(1-e^{-\lambda t_w})(1-P_{md}),~~~~~\,
  \end{split}
\end{equation}
\noindent and
\begin{equation}
{\text{P}_{14}}=1-{\text{P}_{12}}.
\end{equation}}
{When the  UE  is at ${\text{S}_{2}}$, it decodes the packet for a duration of $t_s$, and  if the next packet is received before the end of the current service time, the UE starts decoding the new packet at the end of service time; otherwise, it moves to ${\text{S}_{3}}$. Therefore,  ${\text{P}_{22}}$ and ${\text{P}_{23}}$ can be obtained as
 \begin{equation}
{\text{P}_{22}}=\Pr[t_{p}\leq t_s| \text{S}_2]=1-e^{-\lambda t_s},
\end{equation}
\noindent and
 \begin{equation}
{\text{P}_{23}}=1-{\text{P}_{22}}. 
\end{equation}}
{When the  UE  is at ${\text{S}_{3}}$, it moves to ${\text{S}_{2}}$ if the next packet is received before the expiry of  $t_i$; otherwise, it moves to ${\text{S}_{4}}$. Therefore,  ${\text{P}_{32}}$ and ${\text{P}_{34}}$ can be expressed as
 \begin{equation}
{\text{P}_{32}}=\Pr[t_{p}\leq t_i| \text{S}_3]=1-e^{-\lambda t_i},
\end{equation}
\noindent and
 \begin{equation}
{\text{P}_{34}}=1-{\text{P}_{32}}. 
\end{equation}}
{Finally, at the end of every sleep cycle, the UE decodes PDWCH, and therefore,
\begin{equation}
{\text{P}_{41}}=1.
\end{equation}} 

{\textbf{Steady state probabilities}: The steady state probability that  the UE is at state    ${\text{S}_{k}}$  ($\text{P}_k$) is defined as 
\begin{equation}
{\text{P}_{k}}=\lim_{n\to\infty}  \Pr( S(\tau_n) = \text{S}_k ).
\end{equation}}
{By utilizing the set of balance equations ($\text{P}_{k}=\sum_{l=1}^{4}\text{P}_{l}\text{P}_{lk}$) and the basic sum of probabilities ($\sum_{k=1}^{4}\text{P}_{k}=1$), the $\text{P}_k$'s can be obtained as follows }
{\begin{equation}
\text{P}_{1}=\text{P}_{4}=\frac{\text{P}_{34}\text{P}_{23}}{ 2\text{P}_{34}\text{P}_{23} +\text{P}_{12}(1+\text{P}_{23}) } ,
\end{equation}
\begin{equation}
\text{P}_{2}=\text{P}_{1}\frac{\text{P}_{12}}{\text{P}_{23}\text{P}_{34}},
\end{equation}
\begin{equation}
\text{P}_{3}=\text{P}_{1}\frac{\text{P}_{12}{}}{\text{P}_{34}}.
\end{equation} }

  {\textbf{Holding times}:  The corresponding   holding time for state ${\text{S}_{k}}$ is denoted by  $ \omega_k$,    $k \in\{1,2,3, 4\}$. The holding times for $\omega_1$, $\omega_2$ and $\omega_4$ are constant and given by: $\mathbb{E}[\omega_1]= t_{on} $,   $\mathbb{E}[\omega_2]= t_{s} $ and $\mathbb{E}[\omega_4]=t_{w}-t_{on}$.   However, $\omega_3$ is dependent on the  inter-packet arrival time ($t_p$).   If   a packet arrives before $t_i$, $\omega_3$ is equal to the  inter-packet arrival time, otherwise  $\omega_3$ equals to $t_i$. Therefore, $\omega_3$ can be calculated as a function of $t_p$ as }
\begin{equation}
\omega_3(t_p)= \left\{ \,
\begin{IEEEeqnarraybox}[][c]{l?s}
\IEEEstrut
t_p,&$\text{for}~~ t_p\leq t_i$ , \\
t_i,&$\text{for}~~ t_p>t_i $ .
\IEEEstrut
\end{IEEEeqnarraybox}
\right.
\end{equation}
\noindent Hence, $\mathbb{E}[\omega_3]$ can be  expressed as 
   \begin{equation}
\mathbb{E}[\omega_3]=\int_{0}^{\infty}\omega_3(t)    f_{p}(t)dt=\frac{1-e^{-\lambda t_i }}{\lambda},
\end{equation}
   \noindent where  $f_{p}(t)=\lambda e^{-\lambda t}$ is the probability density function    of the  exponentially distributed packet   arrival time.
   
\begin{figure*}[!t]
\normalsize
\setcounter{equation}{14}
\begin{equation}
\overline{\text{P}}_c=\frac{0.5\text{P}_{1}\text{P}_{12}t_{su}(\text{PW}_2-\text{PW}_4)+0.5\text{P}_{3}\text{P}_{34}t_{pd}(\text{PW}_3-\text{PW}_4)+\sum_{n=1}^{4}\text{P}_{n}\mathbb{E}[\omega_n]\text{PW}_{n}}{\text{P}_{1}\text{P}_{12}t_{su}+\text{P}_{3}\text{P}_{34}t_{pd}+\sum_{n=1}^{4}\text{P}_{n}\mathbb{E}[\omega_n]}, \label{eq:pc_row}
\end{equation}
\hrulefill
\end{figure*}     

\subsection{Average Power Consumption}
The  average power consumption of  the UE, denoted by $\overline{\text{P}}_c$,   can be calculated as the ratio of the average energy consumption and the corresponding overall observation period. It is given by Eq.~\eqref{eq:pc_row} at the top of the next page,
 where $t_{su}$ and $t_{pd}$ correspond to the length of the start-up and power-down stages (transition times), respectively. The corresponding average energy consumption of transitions between  states are calculated as the areas under the power profiles of start-up and power-down stages, see Fig.~\ref{fig:drxvswrx}, whose contribution to the average energy consumption is multiplied by its probability of occurrence ($\text{P}_{1}\text{P}_{12}$ and $\text{P}_{3}\text{P}_{34}$, respectively), thus leading to $0.5\text{P}_{1}\text{P}_{12}t_{su}(\text{PW}_2-\text{PW}_4)$ and $0.5\text{P}_{3}\text{P}_{34}t_{pd}(\text{PW}_3-\text{PW}_4)$, respectively. 

\setcounter{equation}{15}

For modeling simplicity, we assume  that $t_{on}\approx0$, $\text{PW}_{4}\approx0$, $P_{fa}\approx0$, and $P_{md}\approx0$.    Therefore,  Eq. (\ref{eq:pc_row}) can be expanded as  a multivariate function of $t_w$ and $t_i$, {denoted by $\overline{\text{P}}_c(t_w,t_i)$,}  as follows
 {\begin{equation}
  \label{eq:pc_row_2}
  \begin{split}
\overline{\text{P}}_c(t_w,t_i)=\\ \text{PW}_3 \frac{e^{\lambda t_i}(\phi t_s e^{\lambda t_s}+\frac{1}{\lambda})+\frac{1}{2}(\phi t_{su}+t_{pd})-\frac{{1}}{\lambda}}{ e^{\lambda t_i}(t_s e^{\lambda t_s}+\frac{1}{\lambda})+\frac{t_w}{1-e^{-\lambda t_w}}+t_{su}+t_{pd}-\frac{{1}}{\lambda}}.
\end{split}
\end{equation}}

In order to provide more insight into $\overline{\text{P}}_c(t_w,t_i)$,  the instantaneous rate of change  of the power consumption with respect to both $t_w$ and $t_i$ is calculated next. Assuming continuous variables, the partial derivatives of $\overline{\text{P}}_c(t_w,t_i)$ with respect to $t_w$ and $t_i$ are
given by 
{\begin{equation}
  \label{eq:pc_der_t_i}
  \begin{split}
\frac{\partial\overline{\text{P}}_c(t_w,t_i)}{\partial t_i}=\text{PW}_3e^{\lambda t_i}(\lambda\phi t_s e^{\lambda t_s}+1)\times\\\frac{\frac{1}{2}((2-\phi)t_{su}+t_{pd})+\frac{t_w}{1-e^{-\lambda t_w}}}{\big(e^{\lambda t_i}(t_s e^{\lambda t_s}+\frac{1}{\lambda})+\frac{t_w}{1-e^{-\lambda t_w}}+t_{su}+t_{pd}-\frac{{1}}{\lambda}\big)^2}, 
\end{split}
\end{equation}}
\noindent and
{\begin{equation}
  \label{eq:pc_der_t_w}
  \begin{split}
\frac{\partial\overline{\text{P}}_c(t_w,t_i)}{\partial t_w}=\text{PW}_3\big((1+\lambda t_w)e^{-\lambda t_w}-1\big) \times \\
\frac{\big(e^{\lambda t_i}(\phi t_s e^{\lambda t_s}+\frac{1}{\lambda})+\frac{1}{2}(\phi t_{su}+t_{pd})-\frac{{1}}{\lambda}\big)}{\big((1-e^{-\lambda t_w})(e^{\lambda t_i}(t_s e^{\lambda t_s}+\frac{1}{\lambda})+t_{su}+t_{pd}-\frac{{1}}{\lambda})+t_w\big)^2}. 
\end{split}
\end{equation}}

It can be seen from Eq. \eqref{eq:pc_der_t_i} that $\frac{\partial\overline{\text{P}}_c(t_w,t_i)}{\partial t_i}>0$ for all feasible values of $t_w$ and $t_i$. From Eq. (\ref{eq:pc_der_t_w}), we can conclude that $\frac{\partial\overline{\text{P}}_c(t_w,t_i)}{\partial t_w}<0$ due to fact that $(1+\lambda t_w)e^{-\lambda t_w}<1$.

Therefore, the average power consumption
$\overline{\text{P}}_c(t_w,t_i)$ is a  strictly increasing function with respect to $t_i$ at $t_i\geq 0$, and it is a  strictly decreasing function with respect to $t_w$ at $t_w\geq 0$.   As   expected, increasing the wake-up cycle $t_w$ for a fixed $t_i$ can reduce the power consumption. However, by increasing $t_i$ for a fixed $t_w$, the power consumption increases.

 \subsection{Average Buffering Delay}
 We next assume that packets arriving during   ${\text{S}_{4}}$  are  buffered at the gNB until   the UE enters ${\text{S}_{2}}$, thus  causing  buffering delay. Without loss of generality, we assume that the radio access network experiences  unsaturated  traffic conditions. {Therefore, all packets that arrive   are served  without any further scheduling delay. }Furthermore, to simplify the delay modeling, we omit the buffering delay caused by packets arriving on ${\text{S}_{1}}$ or at the start-up state of the modem. This is because the additional buffering delay of such packet arrivals is anyway very small ($t_{on}+t_{su}$). Additionally, thanks to the adoption of the WuS seeking to reduce unnecessary start-ups, the number of occurrences of such scenarios is low. 
 
{Due to slot-based frame    structure of NR, where PDCCH is sent at the beginning of the TTI, inherently, all  packet arrivals (regardless    WuS is utilized or not) suffer from small buffering delay. Since Poisson arrivals are independently and uniformly distributed on any short interval, we assume that the arrival instant of the  packet is uniformly distributed within the TTI, and hence an average extra delay of $t_s/2$ will be introduced.}

\begin{figure}[!t]
\centering
\includegraphics[scale=1.1]{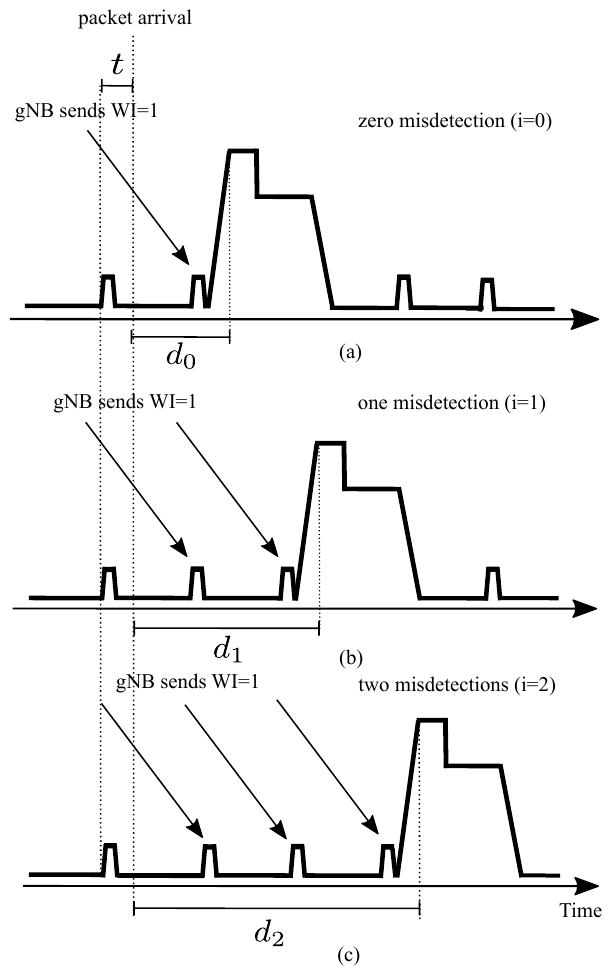}
\caption{Buffering delay caused by wake-up scheme when (a) there is no misdetection,  (b) there is a single misdetection, and (c) there are two   consecutive misdetections. }
\label{fig:delay_md}
\end{figure}

Now, as already briefly mentioned in Section \ref{sec:Background}, misdetections can in general increase the buffering delay. For this purpose, Fig. \ref{fig:delay_md} illustrates  the buffering delay experienced by the UE with no misdetections, with a single misdetection, and with two consecutive misdetections. The number of consecutive misdetections and the   corresponding buffering delay are referred to as $i$ and  $d_i$, respectively, and their dependence on $t$ can be written as $d_i(t)=(i+1)t_w+t_{su}+t_{on}-t$, for $i\in \{0,1,...\}$.  Therefore, the average buffering delay, denoted by $\overline{\text{D}}$, can be expressed as
\begin{eqnarray}
\overline{\text{D}}= \text{P}_4 \sum_{i=0}^{\infty}(1-P_{md}) { (P_{md} )}^{ i}\int_{0}^{{t_{w}}}f_p(t)d_i(t)dt+\frac{t_s}{2}.
\label{eq:d_equation}
\end{eqnarray} 

\begin{figure*}[!t]
\normalsize
\setcounter{equation}{20}
\begin{equation}
  \label{eq:delay_der_t_w}
\frac{\partial\overline{\text{D}}(t_w,t_i)}{\partial t_w}=\frac{e^{\lambda t_i}(1+e^{\lambda t_s})\big(1-(1+\lambda t_w)e^{-\lambda t_w}\big)+2(1-e^{-\lambda t_w})+2\lambda t_{su}e^{-\lambda t_w}}{\Big({2+(1-e^{-\lambda t_w})(1+e^{\lambda t_s})e^{\lambda t_i}}\Big)^2} ,
\end{equation}  
  \begin{equation}
  \label{eq:delay_der_t_i}
\frac{\partial\overline{\text{D}}(t_w,t_i)}{\partial t_i}=\frac{-\lambda e^{\lambda t_i}(1+e^{\lambda t_s})(1-e^{-\lambda t_w})\big(t_w+(t_{su}-\frac{1}{\lambda}){(1-e^{-\lambda t_w})}\big)}{\Big({2+(1-e^{-\lambda t_w})(1+e^{\lambda t_s})e^{\lambda t_i}}\Big)^2} .
\end{equation}
\hrulefill
\end{figure*}     

\setcounter{equation}{19}

Furthermore, due to the small value of misdetection probability ($P_{md}\approx 0$), the contribution of multiple consecutive misdetections to average buffering delay is small. Thus, the average buffering delay for $P_{md}\approx 0$ can be  expanded and solved as a multivariate function of $t_w$ and $t_i$, $\overline{\text{D}}(t_w,t_i)$, as follows
\begin{equation}
\begin{split}
\overline{\text{D}}(t_w,t_i)= \text{P}_4 \int_{0}^{{t_{w}}}f_p(t)d_0(t)dt+\frac{t_s}{2}=\\\frac{t_w+(t_{su}-\frac{1}{\lambda}){(1-e^{-\lambda t_w})}}{2+(1-e^{-\lambda t_w})(1+e^{\lambda t_s})e^{\lambda t_i}}+\frac{t_s}{2}.
\end{split}
\label{eq:d1_equation}
\end{equation}
{We note that $\overline{\text{D}}(t_w,t_i)$ in Eq. \eqref{eq:d1_equation} is strictly-speaking a lower bound of the delay expression in \eqref{eq:d_equation}}.

Similarly to  $\overline{\text{P}}_c(t_w,t_i)$, the partial derivatives of $\overline{\text{D}}(t_w,t_i)$ with respect to  continuous variables   $t_w$ and $t_i$ are given by Eq.~\eqref{eq:delay_der_t_w} and Eq.~\eqref{eq:delay_der_t_i}, respectively, at the top of the page. 

\setcounter{equation}{22}
From Eq. (\ref{eq:delay_der_t_w}), it can be easily concluded that $\frac{\partial\overline{\text{D}}(t_w,t_i)}{\partial t_w}>0$, due to fact that $\big(1-(1+\lambda t_w)e^{-\lambda t_w}\big)>0$. Moreover, it can be shown that  $\frac{\partial\overline{\text{D}}(t_w,t_i)}{\partial t_i}<0$ as follows 
   \begin{equation}
    \begin{split}
\lambda t_w >0\Longrightarrow \lambda t_w >1-e^{-\lambda t_w}\Longrightarrow \\
\frac{\lambda t_w}{1-e^{-\lambda t_w}} >1\Longrightarrow\frac{\lambda t_w}{1-e^{-\lambda t_w}} >1-\lambda t_{su}\Longrightarrow \\
 t_w+(t_{su}-\frac{1}{\lambda}){(1-e^{-\lambda t_w})} >0\Longrightarrow \frac{\partial\overline{\text{D}}(t_w,t_i)}{\partial t_i}<0.
 \end{split}
\end{equation}

Therefore, the average buffering delay 
$\overline{\text{D}}(t_w,t_i)$ is a  strictly increasing function with respect to $t_w$ at $t_w\geq 0$, and it is a  strictly decreasing function with respect to $t_i$ at $t_i\geq 0$. As expected, contrary to  the behavior of $\overline{\text{P}}_c(t_w,t_i)$, increasing the wake-up cycle $t_w$ for a fixed $t_i$ increases the buffering delay. On the other hand, by increasing $t_i$ for a fixed $t_w$, the buffering delay can be reduced. 

{The findings related to the impact of $t_w$ and  $t_i$ on the average delay and power consumption are intuitive while are rigorously confirmed and quantified by the presented expressions.}
\section{Optimization Problem   Formulation and Solution}
\label{sec:prob}
In this section, dual-parameter ($t_w$ and $t_i$)  
constrained  optimization
problem is formulated with the objective of  minimizing  the UE power consumption under a buffering  delay  constraint.    Specifically, the average buffering delay is constrained to be less than or equal to a predefined maximum tolerable delay or delay bound, denoted by  ${\overline{\text{D}}}_{\max}$,  whose value is set based on the service type. To this end, building on the modeling results of the previous section, the optimization problem is now formulated as follows 
\begin{eqnarray}
\mathop{\text{minimize}}_{t_w,t_i} && \overline{\text{P}}_c(t_w,t_i) \label{eq:i1}  \ \ \\ 
\text{subject to} &&  \overline{\text{D}}(t_w,t_i)\leq \overline{\text{D}}_{\max},
 \label{eq:c1} \\
 && t_{w}, t_i \in \{1,2,...\},
\label{eq:c6}
\end{eqnarray}
\noindent  {where $\overline{\text{P}}_c(t_w,t_i)$ and $\overline{\text{D}}(t_w,t_i)$ are defined in Eq.~\eqref{eq:pc_row_2} and Eq.~\eqref{eq:d1_equation}, respectively.}
 
The resulting  optimization problem in (\ref{eq:i1})-(\ref{eq:c6}) belongs to a class of intractable  {mixed-integer non-linear programming (MINLP)} problems \cite{mnlp}. In this work, the corresponding  {MINLP} is solved by using the equivalent  {non-linear programming problem}  with continuous variables, expressed below in (\ref{eq:i2})-(\ref{eq:c7}), which is obtained by means of relaxing the second constraint  (\ref{eq:c6}) into a continuous constraint  (see Eq. (\ref{eq:c7})), assuming that both parameters  are positive real numbers larger than or equal to one (i.e., the minimum TTI unit).  The relaxed optimization problem can be expressed as 
\begin{eqnarray}
\mathop{\text{minimize}}_{t_w,t_i} && \overline{\text{P}}_c(t_w,t_i) \label{eq:i2}  \ \ \\ 
\text{subject to} &&  \overline{\text{D}}(t_w,t_i)\leq \overline{\text{D}}_{\max},
 \label{eq:c2} \\
&& t_{w} \geq 1 , t_i \geq 1.
\label{eq:c7}
\end{eqnarray}
In general, the optimization problem in \eqref{eq:i2}-\eqref{eq:c7} is not {jointly convex in} $t_w$ and $t_i$. Therefore, finding the global optimum is a challenging task. However, in the next subsections, we exploit the increasing/decreasing properties of the power consumption and delay expressions that we have derived in   Section \ref{sec:sysmodel}, in order to derive additional properties of the problem   that will allow us to find the optimal solution in closed form.

 \subsection{Unbounded  Feasible Region}
 \label{sec:feasib}
In this section, a schematic approach is used to illustrate the feasible region for the relaxed optimization problem in (\ref{eq:i2})-(\ref{eq:c7}) and then  the feasible region is narrowed down  to the boundary of the delay constraint, whose points  are proved to remain candidate solutions while the other feasible solutions are henceforth excluded.

{Fig. \ref{fig:she_proof} (a) and (b)} show the increasing trend of the power consumption and the decreasing behaviour of the delay constraint as a function of $t_i$, while $t_w$ is fixed at $t_{w_0}$, i.e.  $\frac{\partial\overline{\text{P}}_c(t_w,t_i)}{\partial t_i}>0$ and $\frac{\partial\overline{\text{D}}(t_w,t_i)}{\partial t_i}<0$ (as proved in Section \ref{sec:sysmodel}). Let us consider  an arbitrary point $A$ in the interior of the feasible region ($t_{i_A} > t_{i_m}$, where  $\overline{\text{D}}(t_{w_0},t_{i_m})=\overline{\text{D}}_{\max}$). As it can be seen from {Fig. \ref{fig:she_proof} (a) and (b)}, there is always a point on the boundary of the delay constraint, referred to as $B$ ($t_{i_B}=t_{i_m}$), where its power consumption  $\overline{\text{P}}_{c_B}$ is lower than that of $A$ ($\overline{\text{P}}_{c_B}<\overline{\text{P}}_{c_A}$). Hence, we can conclude that, for any fixed $t_w$, under a given delay constraint, there is a point on the boundary that attains the lowest power consumption.

Similarly,  {Fig. \ref{fig:she_proof}  (c) and (d)} show the decreasing trend of the power consumption and increasing behaviour of the delay constraint as a function of $t_w$, while $t_i$ is fixed at $t_{i_0}$, i.e.  $\frac{\partial\overline{\text{P}}_c(t_w,t_i)}{\partial t_w}<0$ and $\frac{\partial\overline{\text{D}}(t_w,t_i)}{\partial t_w}>0$ (as proved in Section \ref{sec:sysmodel}).  Consider  an arbitrary point $C$ in the interior of the feasible region ($t_{w_C} < t_{w_m}$  where  $\overline{\text{D}}(t_{w_m},t_{i_0})=\overline{\text{D}}_{\max}$). As it can be seen from {Fig. \ref{fig:she_proof} (c) and (d)}, there is always a point on the boundary of the delay constraint, referred to as $D$ ($t_{w_D}=t_{w_m}$), where its power consumption  $\overline{\text{P}}_{c_D}$ is lower than that of $C$ ($\overline{\text{P}}_{c_D}<\overline{\text{P}}_{c_C}$). Then, we can conclude that, for any fixed $t_i$, under a given delay constraint, there is a point on the boundary that attains the lowest power consumption.

Therefore, because for both scenarios (fixed  $t_w$  and fixed  $t_i$), the lowest power consumption occurs at the boundary of the delay constraint, we can conclude that the optimal point cannot be located in the interior of the feasible region, but rather it lies over the boundary. That is, any  arbitrary point $(t_w,t_i)$ in the feasible region  of the relaxed optimization problem in  (\ref{eq:i2})-(\ref{eq:c7}) cannot be an optimal point, unless it lies on the boundary (rather than the interior) of the delay constraint, i.e., $\overline{\text{D}}(t_w,t_i)= \overline{\text{D}}_{\max}$.

\begin{figure}
\centering
\includegraphics[width=0.50  \textwidth]{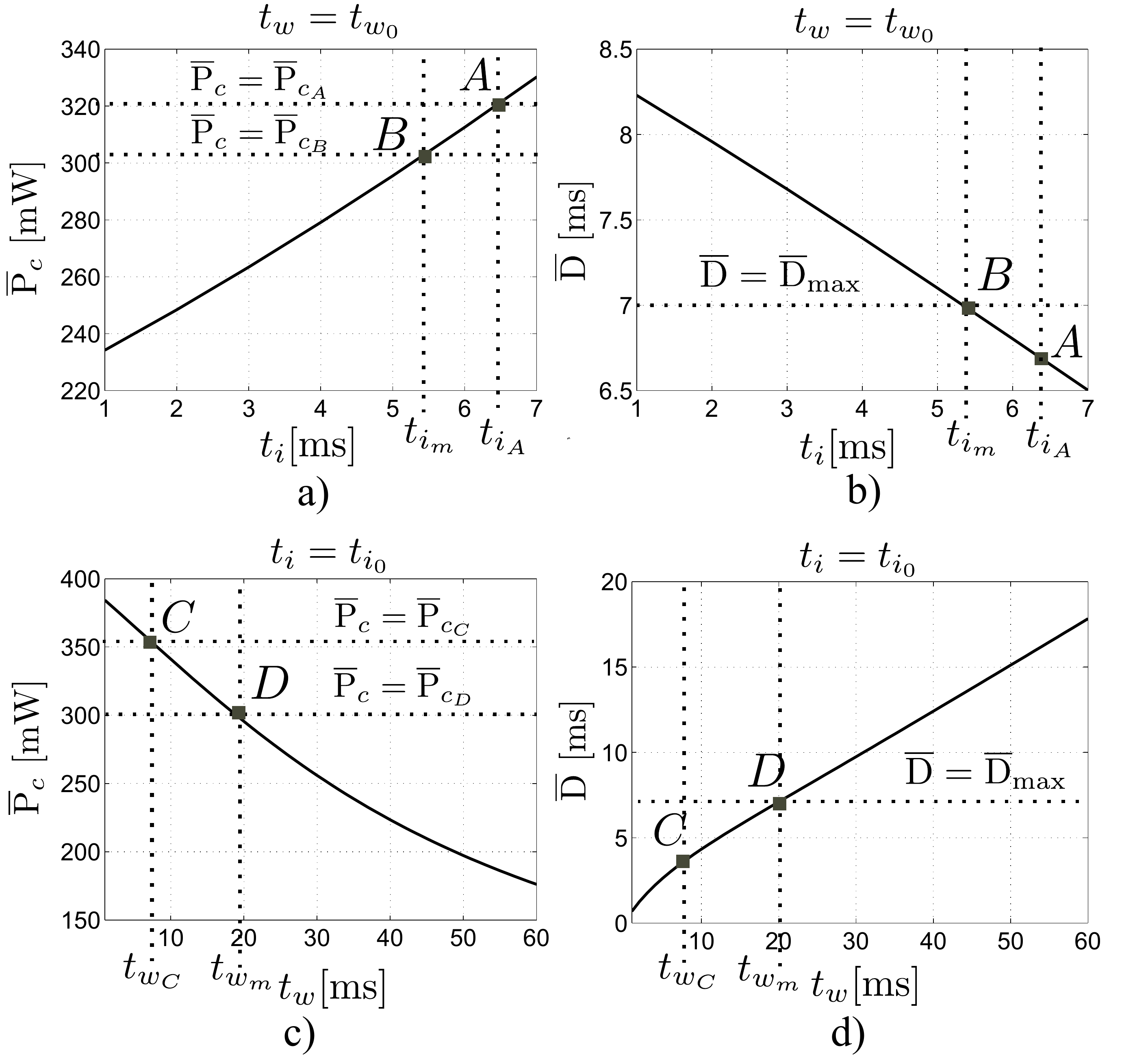}
\caption{Schematic proof that optimal point lies over the boundary.}
\label{fig:she_proof}
\end{figure}
\subsection{Power Consumption over Boundary}
\label{sec:pcob}
Next, the equation of the boundary curve (expressed through   $t_i$ as a function of $t_w$)  is derived, and then the power consumption profile of all points on the boundary  is calculated as well as formulated as a function of $t_w$ only. In particular, the boundary curve can be obtained  by finding all the solutions for which the  inequality constraint in (\ref{eq:c2}) is satisfied with equality, while the constraint in (\ref{eq:c7}) is met, i.e., 
  \begin{align}
 \overline{\text{D}}(t_w,t_i)=\overline{\text{D}}_{\max} \text{~for all~}  t_{w}, t_i \geq 1.
 \label{eq:c22}
\end{align}
 
By utilizing  Eq. (\ref{eq:d1_equation}), we can isolate $t_i$, and the boundary curve can be formulated  as follows 
 \begin{equation}
    \begin{split}
t_i(t_w)=\frac{1}{\lambda}\ln \Big(\frac{{t_w+(t_{su}-\frac{1}{\lambda}){(1-e^{-\lambda t_w})}}-2(\overline{\text{D}}_{\max}-\frac{t_s}{2})}{(\overline{\text{D}}_{\max}-\frac{t_s}{2})(1-e^{-\lambda t_w})(1+e^{\lambda t_s})}\Big) \\\text{for all~}  t_{w}\geq t_{w_b},
 \label{eq:bn}
   \end{split}
\end{equation}
\noindent where $t_{w_b}$ (see Eq. (\ref{eq:twb})) is the minimum feasible value of  $t_{w}$ over the boundary.

{By using Eq. (\ref{eq:bn}), one can show that $t_i(t_w)$ is an increasing function with respect to $t_w$ on any  feasible $t_w$ point  over the boundary of the delay constraint (i.e., $\frac{d t_i(t_w)}{d t_w} \geq 0$), as follows.} Let us use the composite function rule over (\ref{eq:bn}), so that $\frac{d t_i(t_w)}{d t_w}$ can be calculated as follows 
  \begin{equation}
 \frac{d t_i(t_w)}{d t_w} =\frac{1}{\lambda}\frac{d \ln (\text{Arg})}{d \text{Arg}}\frac{d \text{Arg}(t_w)}{d t_w},
 \label{eq:part}
   \end{equation} 
   
   \noindent where $\text{Arg}$   refers to  the  argument of the logarithm in  (\ref{eq:bn}).  Since the logarithmic  function is monotonically increasing ($\frac{d \ln(\text{Arg})}{d \text{Arg}} \geq 0$), it is sufficient to   prove that  $\text{Arg}(t_w)$  is increasing with respect to  $t_w$,   
{ \begin{equation}
\frac{d \text{Arg}(t_w)}{d t_w}=\frac{1-(1-2(\overline{\text{D}}_{\max}-\frac{t_s}{2})\lambda +\lambda t_w)e^{-\lambda t_w}}{(\overline{\text{D}}_{\max}-\frac{t_s}{2})(1+e^{\lambda t_s})(1-e^{-\lambda t_w})^2},
\end{equation}
\noindent from which we can write
   \begin{equation}
    \begin{split}
1-(1+\lambda t_w)e^{-\lambda t_w}>0\Longrightarrow \\
1-(1-2(\overline{\text{D}}_{\max}-\frac{t_s}{2})\lambda +\lambda t_w)e^{-\lambda t_w}  \geq 0 \Longrightarrow \\
 \frac{d \text{Arg}(t_w)}{d t_w} \geq 0.
 \end{split}
 \label{eq:part1}
\end{equation}}
\noindent Therefore, we can conclude that $\frac{d t_i(t_w)}{d t_w} \geq 0$.

Additionally, one can prove that  $t_i=1$ and $t_w=t_{w_b}$ (in which  $t_{w_b}$ always exists, and is larger than or equal to one) is located over the boundary, as follows.  Based on Eq. (\ref{eq:bn}), and  by induction on  $t_i=1$, we can write 
 {   \begin{equation}
    \begin{split}
e^{\lambda}=\frac{{t_{w}+(t_{su}-\frac{1}{\lambda}){(1-e^{-\lambda t_{w}})}}-2(\overline{\text{D}}_{\max}-\frac{t_s}{2})}{(\overline{\text{D}}_{\max}-\frac{t_s}{2})(1-e^{-\lambda t_{w}})(1+e^{\lambda t_s})} \Longrightarrow \\
e^{-\lambda t_{w}}=-\frac{-\lambda t_{w}    {+\big((e^{\lambda}(1+e^{\lambda t_s})+2)\overline{\text{D}}_{\max} -t_{su}\big)\lambda+1}}{\lambda t_{su}-e^{\lambda}(\overline{\text{D}}_{\max}-\frac{t_s}{2})(1+e^{\lambda t_s})\lambda-1} \Longrightarrow \\
e^{-\lambda t_{w}+F}=-e^{F} \frac{-\lambda t_{w} +F}{H}\Longrightarrow \\
{-\lambda t_{w}+F}=- \mathcal W\Big(\frac{H}{e^{F}}\Big)\Longrightarrow \\
 t_{w} = \frac{1}{\lambda} \Big(F+\mathcal W\Big(\frac{H}{e^{F}}\Big)\Big),
 \label{eq:sss}
    \end{split}
      \end{equation}}
      
\noindent where  {$F= \big((e^{\lambda}(1+e^{\lambda t_s})+2)(\overline{\text{D}}_{\max}-\frac{t_s}{2}) -t_{su}\big)\lambda+1$, $H=\lambda t_{su}-e^{\lambda}(\overline{\text{D}}_{\max}-\frac{t_s}{2})(1+e^{\lambda t_s})\lambda-1$} and $\mathcal W(x)$ is the Lambert W function \cite{Corless1996}.  For typical $\overline{\text{D}}_{\max}$ and  $t_{su}$ values,      $1 \ll \frac{F}{\lambda}$  and $H<0$, therefore, the main branch of the Lambert W function ($\mathcal W_0$) can be considered as a solution for (\ref{eq:sss}) that has a value greater than $-1$. Then,  we can conclude that   
  \begin{equation}
 \label{eq:twb}
 t_{w_b}=\frac{1}{\lambda} \Big(F+\mathcal W_0\Big(\frac{H}{e^{F}}\Big)\Big)\geq 1,
  \end{equation}
\noindent and, as a result, all feasible points on the boundary curve  can be specified and constrained  by     $t_i\geq 1$ and $t_w \geq t_{w_b}$.

The $t_{w_b}$ in \eqref{eq:twb} is the smallest feasible $t_w$ over the boundary curve because if we assume that there is a $t_w$ smaller than $t_{w_b}$, based on Eq. (\ref{eq:part}) and (\ref{eq:part1}), its corresponding $t_i$ should become  smaller than one, which belongs to the unfeasible region. Therefore, based on proof-by-contradiction, ($t_i=1$, $t_w=t_{w_b}$) lies over the corner part of the boundary. Consequently, ${t_i \geq1}$ and ${t_w \geq t_{w_b}}$ are 
equivalent  constraints of the boundary of the delay constraint. Therefore, the point ($t_i=1$, $t_w=t_{w_b}$) is an extreme point, and it is  located over the boundary curve of the delay constraint, where $t_{w_b}$ is larger than or equal to one.

Finally, by substituting the value of $e^{\lambda t_i}$  over the boundary  (argument of logarithm in  (\ref{eq:bn}))   into   (\ref{eq:pc_row_2}),  the average power consumption of all the points over the boundary, referred to as $\overline{\text{P}}_b(t_w)$,  can be obtained as follows 
 {   \begin{equation}
    \begin{split}
\overline{\text{P}}_{b}(t_w)=\text{PW}_3\frac{u_1+u_2 t_w+u_3e^{-\lambda t_w}}{w_1+w_2t_w+w_3e^{-\lambda t_w}} \text{~~~~for all~}  t_{w}\geq t_{w_b}, 
 \label{eq:xc}
   \end{split}
\end{equation}}
\noindent where
 {   \begin{align} 
 \begin{split}
u_1&=\big (\frac{1}{2}(\phi t_{su}+t_{pd})-\frac{{1}}{\lambda}\big )(1+e^{\lambda t_s})(\overline{\text{D}}_{\max}-\frac{t_s}{2})+\\ & \quad  (t_s e^{\lambda t_s}+\frac{1}{\lambda})\big(t_{su}-\frac{1}{\lambda}-2(\overline{\text{D}}_{\max}-\frac{t_s}{2})\big),\end{split}\\
\label{eq:Ab}
u_2&=t_s e^{\lambda t_s}+\frac{1}{\lambda},\\
 \begin{split}
u_3&= -\big (\frac{1}{2}(\phi t_{su}+t_{pd})-\frac{{1}}{\lambda}\big )(1+e^{\lambda t_s})(\overline{\text{D}}_{\max}-\frac{t_s}{2})-\\ & \quad  (t_s e^{\lambda t_s}+\frac{1}{\lambda})\big(t_{su}-\frac{1}{\lambda}\big), \end{split} \\ \begin{split} 
w_1&=\big (\frac{1}{2}(t_{su}+t_{pd})-\frac{{1}}{\lambda}\big )(1+e^{\lambda t_s})(\overline{\text{D}}_{\max}-\frac{t_s}{2})+\\ & \quad  (t_s e^{\lambda t_s}+\frac{1}{\lambda})\big(t_{su}-\frac{1}{\lambda}-2(\overline{\text{D}}_{\max}-\frac{t_s}{2})\big), \end{split}   \\
w_2&=t_s e^{\lambda t_s}+\frac{1}{\lambda}+(1+e^{\lambda t_s})(\overline{\text{D}}_{\max}-\frac{t_s}{2}),   \\
\begin{split}
w_3&= -\big (\frac{1}{2}(t_{su}+t_{pd})-\frac{{1}}{\lambda}\big )(1+e^{\lambda t_s})(\overline{\text{D}}_{\max}-\frac{t_s}{2})-\\ & \quad  (t_s e^{\lambda t_s}+\frac{1}{\lambda})\big(t_{su}-\frac{1}{\lambda}\big).\end{split}  
\end{align}}
In the Appendix, we further analyze the expression in Eq.~(\ref{eq:xc}) in detail.

\subsection{Optimal Parameter Values}
The power consumption  over the boundary curve in Eq. (\ref{eq:xc}) depends on the packet arrival rate $\lambda$. Furthermore, as it is shown in the Appendix,  $\overline{\text{P}}_{b}(t_w)$   behaves differently for different ranges of $\lambda$. For this purpose, $\frac{d \overline{\text{P}}_{b}(t_w)}{d t_w}$ is calculated  (a detailed analysis is provided in the Appendix). Briefly, its sign for different ranges of $\lambda$  within {the feasible region of the wake-up cycle (i.e., $t_{w_b} \leq t_w $)} can be expressed as follows 
{ \begin{equation}
\text{sgn}\left(\frac{d \overline{\text{P}}_{b}(t_w)}{d t_w}\right)= \left\{ \,
\begin{IEEEeqnarraybox}[][c]{l?s}
\IEEEstrut
1 &      $\text{for}  ~ 0< \lambda \leq {\lambda}_t$ ,  \\
-1 & $\text{for} ~ {\lambda}_t< \lambda < 1$ ,
\IEEEstrut
\end{IEEEeqnarraybox}
\right.
\label{eq:ifbwzz}
\end{equation}}

\noindent where sgn(.) refers to the sign function, and $\lambda_t$ is referred to as the turnoff packet arrival rate. {The turnoff packet arrival rate can be calculated  using any typical root-finding algorithm that meets $F_1=0$   (see details in Appendix) where}
  {   \begin{align}
  F_1&=(1+e^{\lambda t_s})(\overline{\text{D}}_{\max}-\frac{t_s}{2})(t_s e^{\lambda t_s}+\frac{{1}}{\lambda}) \times   \\
& \quad  \big(\frac{1}{2}(t_{pd}-\phi t_{su})+\frac{{1}}{\lambda}+2(\overline{\text{D}}_{\max}-\frac{t_s}{2})\big)- \\
& \quad \big (\frac{1}{2}(\phi t_{su}+t_{pd})-\frac{{1}}{\lambda}\big )(1+e^{\lambda t_s})^2(\overline{\text{D}}_{\max}-\frac{t_s}{2})^2.
  \end{align}}
\begin{theorem}
$t_w^*= t_{w_b}~ \text{and} ~ t_i^*=1$ are the optimal parameter values  of the  optimization problem in (\ref{eq:i2})-(\ref{eq:c7}) for the range  $0< \lambda \leq {\lambda}_t$. 
\end{theorem} 

 \begin{proof}
As it can be seen in (\ref{eq:ifbwzz}),   for all $0< \lambda \le {\lambda}_t$, the power consumption increases when increasing $t_w$ over the boundary, so that the minimum power consumption is achieved at the minimum feasible $t_w$, i.e.,  $t_w=t_{w_b}$. Correspondingly, the optimal value of $t_i$  can be calculated by substituting $t_{w_b}$ into Eq. (\ref{eq:bn}), which leads to    $t_i=1$. Therefore, for all $0<\lambda\leq{\lambda}_t$,  $t_w^*=t_{w_b}$ and $t_i^*=1$ is the optimal solution of the optimization problem in (\ref{eq:i2})-(\ref{eq:c7}).
\end{proof}
  
{\begin{theorem}
$t_w^*=+\infty~\text{and}~t_i^*=+\infty$ are the optimal parameter values  of the  optimization problem in (\ref{eq:i2})-(\ref{eq:c7}) for the range  ${\lambda}_t< \lambda <1$. 
\end{theorem} }
{\begin{proof}
As it can be seen in (\ref{eq:ifbwzz}),   for all ${\lambda}_t< \lambda <1$, the power consumption decreases when increasing $t_w$ over the boundary, so that the minimum power consumption is achieved at the maximum feasible $t_w$, i.e.,  $t_w=+\infty$. Correspondingly, the optimal value of $t_i$  can be calculated by substituting $t_{w}$ into Eq. (\ref{eq:bn}), which is    $t_i=+\infty$. Therefore, for all ${\lambda}_t< \lambda <1$,  $t_w^*=+\infty$ and $t_i^*=+\infty$ is the optimal solution of the optimization problem in (\ref{eq:i2})-(\ref{eq:c7}).
 \end{proof}}

 {
 \begin{corollary}
For the range  ${\lambda}_t< \lambda <1$, the optimal solution  is equivalent to not utilizing WuS; the system is always at active-decoding and active-inactivity timer states (P$_{1}+$P$_{4}\approx 0$). Hence, if the energy and delay overhead of switching on/off the BBU are taken into account, the WuS is not effective anymore for high $\lambda$ values. Instead, other power saving mechanisms, such as DRX, microsleep, or pre-grant message  could be used in this regime.  
 \end{corollary}}

{Fig.  \ref{fig:lambda_t} (a) and (b)} illustrate how ${\lambda}_t$  changes with $\overline{\text{D}}_{\max} $ and $t_{su}+t_{pd}$, respectively. As it can be observed in {Fig. \ref{fig:lambda_t} (a)}, the turnoff packet arrival rate is independent and insensitive  to variations of the value of  $\overline{\text{D}}_{\max}$, however, it reduces, when   $t_{su}+t_{pd}$ becomes larger (see {Fig. \ref{fig:lambda_t} (b)}).  Therefore, in order to decide whether to enable WuS or not, regardless of the QoS requirement of the considered traffic, the network needs to compare the estimated packet arrival rate with pre-calculated and fixed ${\lambda}_t$.

\begin{figure}
\centering
\includegraphics[width=0.5  \textwidth]{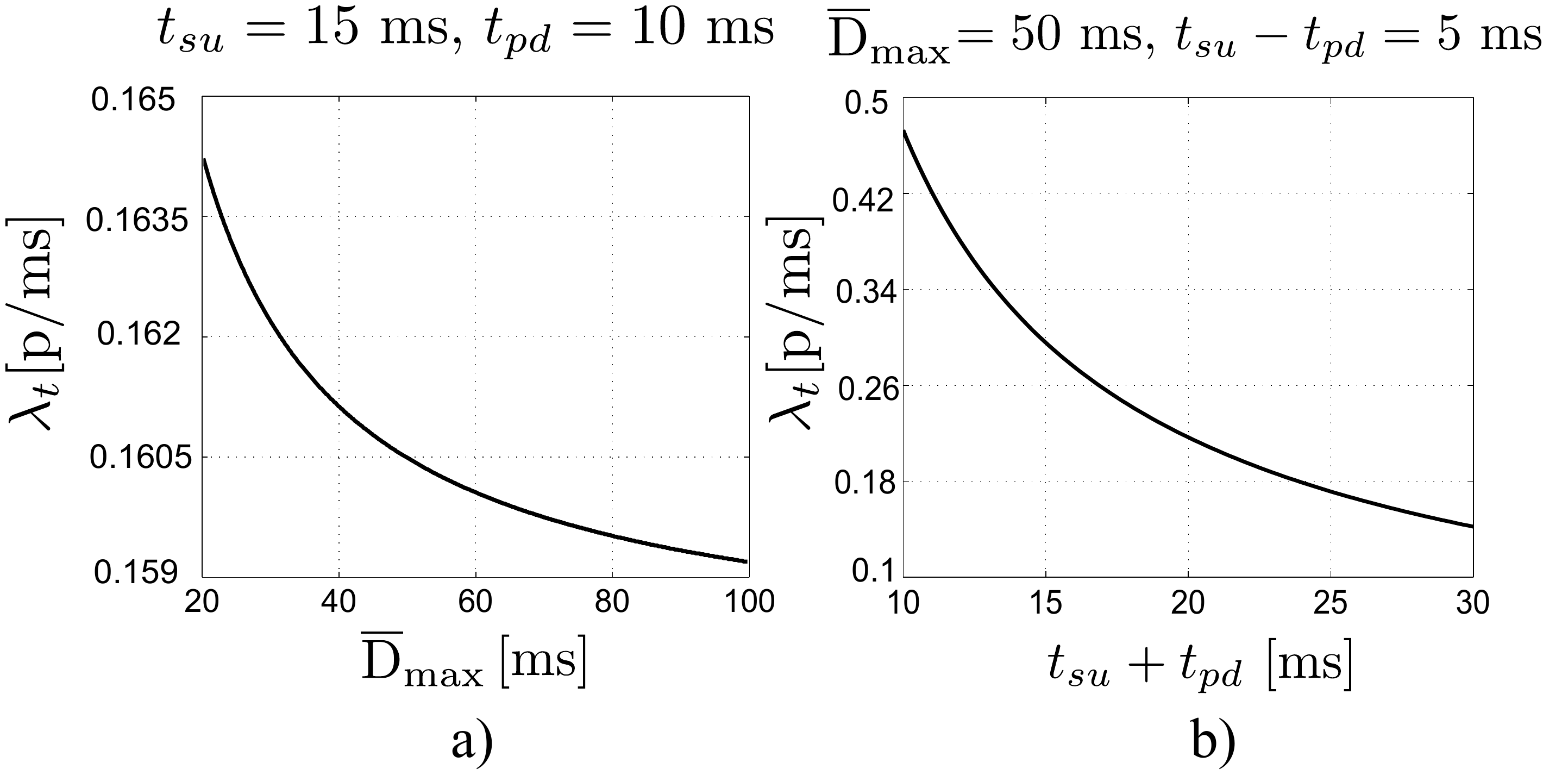}
\caption{{${\lambda}_t$  as function of delay bound and sum of transition times.}}
\label{fig:lambda_t}
\end{figure}

Interestingly,  for $ {\lambda}_t< \lambda < 1$,  the power consumption reduces by increasing $t_w$ {towards infinity (and correspondingly, $t_i$ increases), while the delay constraint is satisfied. This can be interpreted in a way that for packet arrival rates higher than ${\lambda}_t$, the WuS is not effective  anymore and only adds overhead energy consumption, thus implying that it is better not to switch off the BBU and to utilize  short DRX cycles. As it is shown in {Fig. \ref{fig:lambda_t} (b)}, when   $t_{su}+t_{pd}$ becomes larger, turnoff packet arrival rates become smaller, which justifies the fact that for the higher packet arrival rates, the frequent start-up and power-down related energy consumption becomes larger.}

{Additionally, the main reason for such interpretation is due to the fact that for large wake-up cycles, P$_{12}$ approaches to one, which is equivalent to the reduction of  number of potential scheduled PDCCHs, and hence there is no  gain by using the wake-up scheme over DRX anymore. Furthermore,  for higher packet arrival rates, based on the optimal policy, most of the time the BBU is at either S$_2$ or S$_3$ (P$_{1}=$P$_{4}\approx 0$), to avoid wasted energy of start-up and power-down times, and to satisfy the delay constraint (illustrated in Fig. \ref{fig:small_large_lambda} (a)).    As it can be seen in    Fig.  \ref{fig:small_large_lambda} (a), for small packet arrival rates, there is considerable energy consumption for transition between states, however once the packet arrival rate is higher than $\lambda_t$, this trend changes and the UE is mainly at S$_2$ or S$_3$, and it does not waste energy in start-up/power-down stages, due to the need for frequent start-up/power-down of the BBU. Such change in behaviour of the wake-up scheme can be explained by the objective of the system which is  to reduce  the overall power consumption, as it shown in Fig. \ref{fig:small_large_lambda} (b). Moreover, as it can be seen in Fig.  \ref{fig:small_large_lambda} (a), for packet arrival rates higher than the turnoff packet arrival rate, most of the energy is consumed for decoding of the packets, and its energy consumption increases linearly with the packet arrival rate.}      
\begin{figure}[!t]
\centering \
\includegraphics[width=0.34 \textwidth]{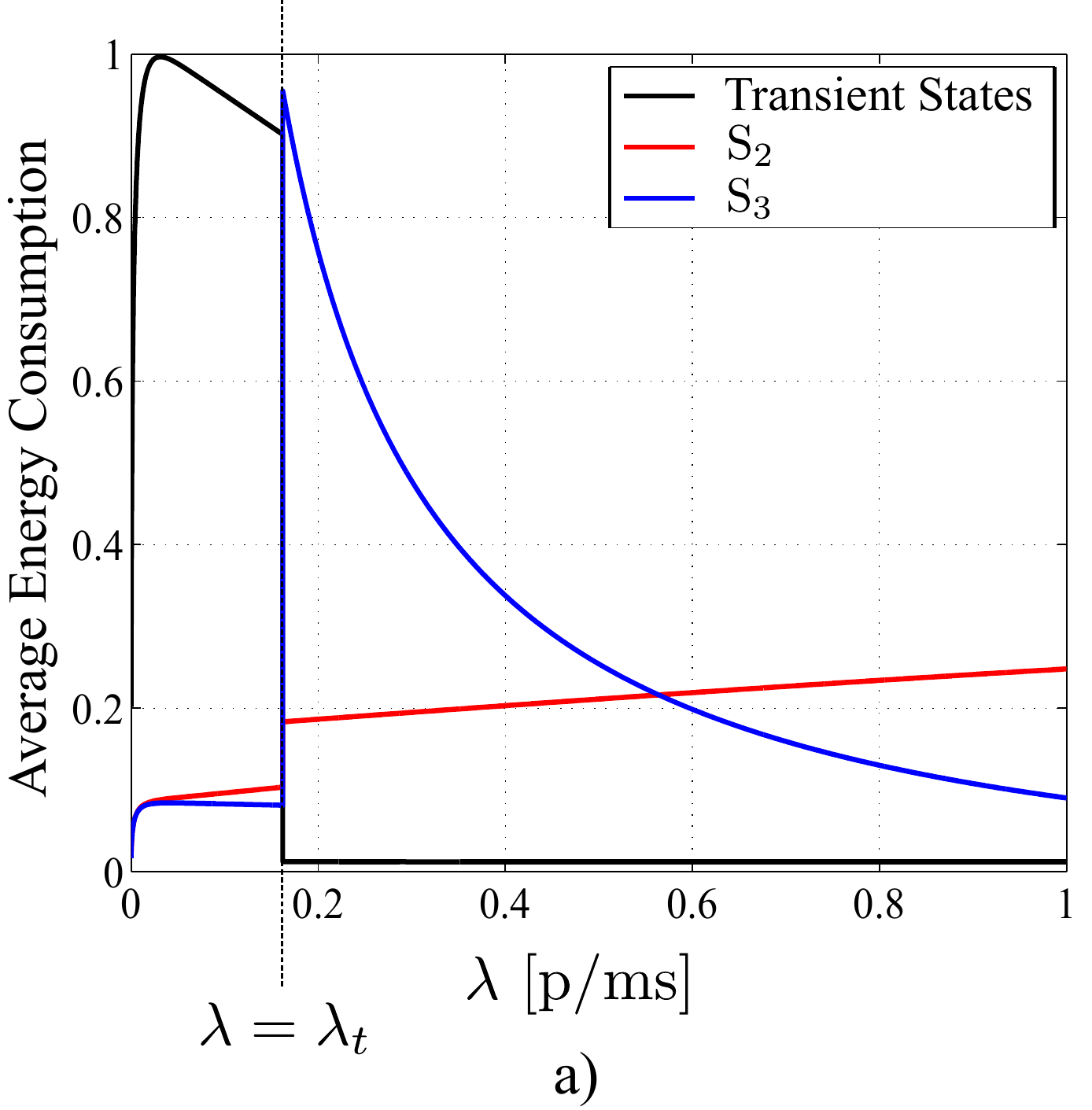}  \includegraphics[width=0.35 \textwidth]{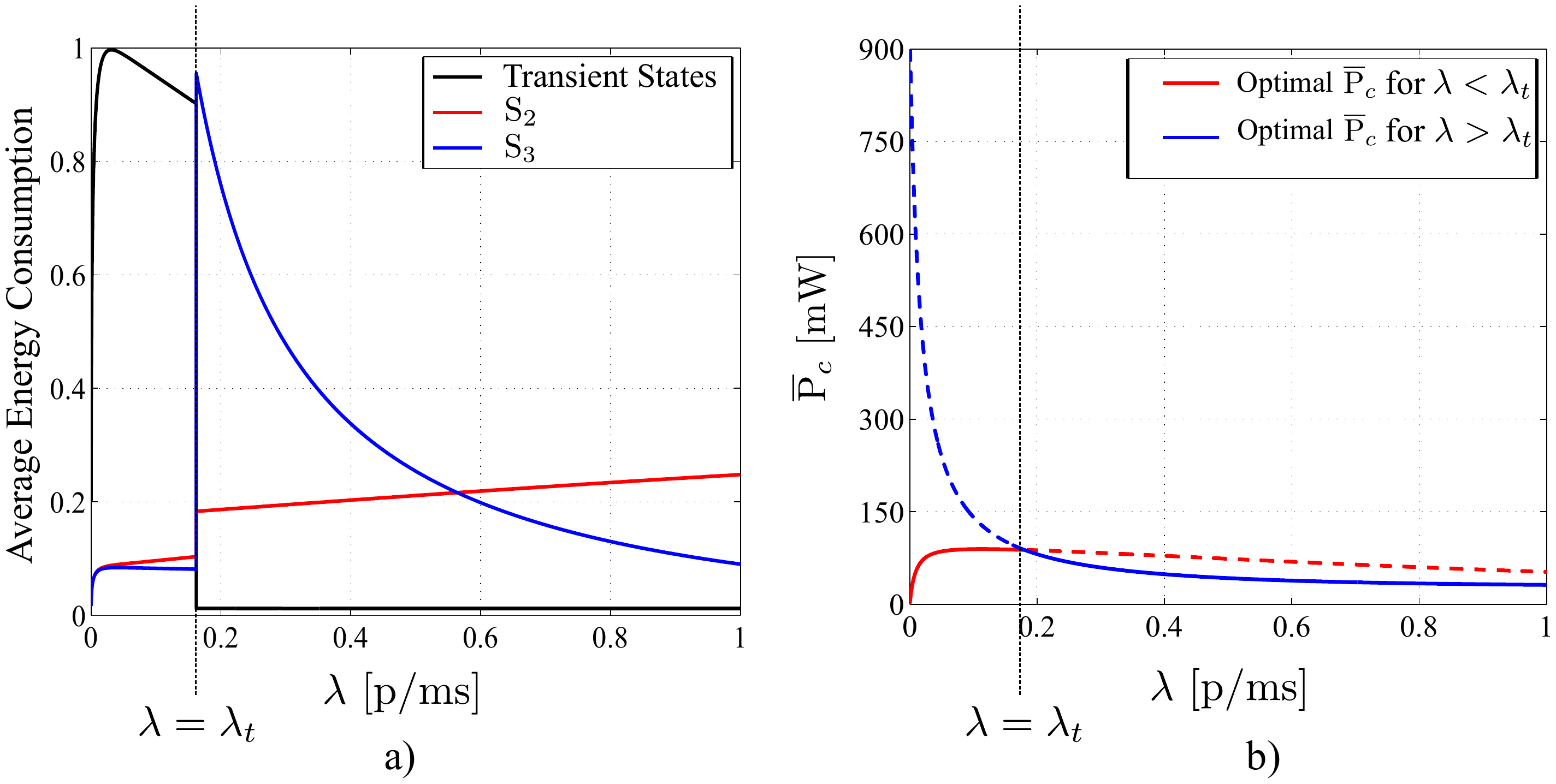} 
\caption{{a) Normalized average energy consumption of each state, as result of optimal configuration of wake-up scheme parameters; b) average overall power consumption of optimized wake-up scheme.}}
\label{fig:small_large_lambda}
\end{figure}

Finally, it can be shown that for $\lambda$ less than the turnoff packet arrival rate ($0< \lambda \le {\lambda}_t$), the optimal parameter values ($t_w^*$ and $t_i^*$) of the original MINLP (\ref{eq:i1})-(\ref{eq:c6}) can be written based on the optimal values of the equivalent relaxed problem (\ref{eq:i2})-(\ref{eq:c7}) as follows 
  \begin{align}
t_w^*=\lfloor t_{w_b} \rfloor~ \text{and} ~ t_i^*=1,
 \label{eq:A}
\end{align}
\noindent where $\lfloor  ~  \rfloor$ refers to the floor function.

\begin{theorem}
$t_w^*=\lfloor t_{w_b} \rfloor~ \text{and} ~ t_i^*=1$ are the optimal parameter values  of the  MINLP in \eqref{eq:i1}-\eqref{eq:c6} for the range  $0< \lambda \leq {\lambda}_t$. 
\end{theorem} 

 \begin{proof}
If $f$ is a   function of continuous variables  $x$ and $y$, it can easily be shown that 
 {  \begin{equation}
    \begin{split}
\text{if}~ \frac{\partial f}{\partial x}>0 \Longrightarrow \Delta {f}_x > 0, ~\forall x\in \mathbb{R}^+\\
\text{if}~ \frac{\partial f}{\partial x}<0 \Longrightarrow \Delta {f}_x<0,  ~\forall x\in \mathbb{R}^+
 \label{eq:Arrrr}
 \end{split}
\end{equation} }
\noindent where 
  { \begin{align}
\Delta {f}_x\triangleq f(\lfloor x \rfloor +1)-f(\lfloor x \rfloor) .
 \label{eq:Az}
\end{align}}
Therefore, by assuming $f$  as representative of either $\overline{\text{P}}_{c}(t_w,t_i)$ or  $\overline{\text{D}}(t_w,t_i)$   and $x$ as of either $t_i$ or $t_w$,  one  can prove, similarly to   the case with continuous variables (proved in  Section \ref{sec:feasib}), that the optimal parameters of MINLP (\ref{eq:i1})-(\ref{eq:c6}) are laid over the boundary.  Therefore, the boundary of (\ref{eq:i1})-(\ref{eq:c6}) consists of all combinations of $(\lfloor t_w \rfloor, t_i)$ for which $ t_i \in \{1,2,...\}$ and $ \overline{\text{D}}(t_w,t_i)= \overline{\text{D}}_{\max}$ (as formulated in (\ref{eq:bn})). Similarly, based on (\ref{eq:Arrrr}) and (\ref{eq:part}), as well as the properties of the floor function, we can state that increasing $t_i$ over the boundary may increase $\lfloor t_w \rfloor$ ($\Delta {\lfloor t_w \rfloor}_{t_i}>0$) or  $\lfloor t_w \rfloor$ can remain in its previous value ($\Delta {\lfloor t_w \rfloor}_{t_i}=0$). Furthermore,  based on (\ref{eq:ifbwzz}), for  $0< \lambda < {\lambda}_t$,  we can conclude that  $\Delta { \overline{\text{P}}_{b}}(t_w,t_i)_{~t_{w}}>0$.  Therefore, $t_w^*$ is the smallest feasible value of the wake-up cycle over the boundary, i.e.,  $t_w^*=\lfloor t_{w_b} \rfloor$. However, $\lfloor t_{w_b} \rfloor$ may correspond to either $t_i=1$ or larger values at the same time. Since the power consumption has the lowest value at the lowest $t_i$, for fixed $t_w$, we can conclude  that  $t_i^*=1$.
\end{proof}

\begin{corollary}
It turns out that the solution of the relaxed problem in \eqref{eq:i2}-\eqref{eq:c7}, after finding the integer part of one of the two optimization parameters' values, are actual optimal values for the original MINLP optimization problem in (\ref{eq:i1})-(\ref{eq:c6}). Therefore, our relaxation approach   yielded an equivalent reformulation.
 \end{corollary}

\section{Numerical Results}
\label{sec:eval}
In this section, a set of numerical results are provided in order to validate our concept and the analytical results, as well as to show and compare the average power consumption of the optimized WuS over DRX for  packet arrival rates less than the turnoff packet arrival rate. Power consumption of the mobile device in different operating states  is highly dependent on the   implementation, and also its operational configurations. Therefore, for   the numerical results, the power consumption model   used in \cite{Globecom, pm2}, \cite{pm1}, \cite{Lauridsen2014} is employed. Its parameters for DRX and WuS are shown in Table  \ref{drxpower} and Table \ref{tab:powr}, respectively. LTE-based power consumption values, shown in Table  \ref{drxpower}, are considered as a practical example since those of the emerging NR modems are not publicly available yet. {For simulations, we use $\phi=1.1$ as an example numerical value, while methodology wise, numerical results can also be generated for any other value as well.}

\begin{table}[!t]
\scriptsize
\renewcommand{\arraystretch}{1.3}
\caption{{ Average transitional time  and representative power consumption values of LTE-based cellular module during short and long DRX when carrier bandwidth is $20$ MHz, and $\phi=1.1$}}
\label{drxpower}
\centering
\scriptsize
\begin{tabular}{|c|c|c|c|c|c|c|c|}
    \hline
     DRX Cycle &$\text{PW}_{\text{sleep}}$& $\text{PW}_{\text{active}}$ &$\text{PW}_{\text{decode}}$ &  $t_{su}$&   $t_{pd}$  \\
    \hline
    \hline
      short&  $395$ mW&    $850$ mW & $935$ mW&  $1$ ms &   $1$ ms\\
          \hline
      long&  $\approx0$ mW&       $850$ mW&$935$ mW&   $15$ ms &   $10$ ms\\
          \hline
\end{tabular}
\end{table} 

\begin{table}[!t]
\scriptsize
\renewcommand{\arraystretch}{1.3}
\caption{{{Assumed power consumption parameters of the wake-up scheme with assumption of $\phi=1.1$ }}}
\label{tab:powr}
\centering
\begin{tabular}{|c||c|c|c|c|c|c|c|c|c|c|}
    \hline
     PW$_{\text{1}}$  & PW$_{\text{2}}$  &PW$_{\text{3}}$  & PW$_{\text{4}}$ &$t_{su}$  & $t_{pd}$& $t_{on}$  \\
    \hline
    \hline
      57mW  &935 mW  &   850 mW  &   $\approx0$ mW&   15 ms  &   10 ms &   1/14 ms \\
    \hline
\end{tabular}
\end{table}

Two different   sets of performance results, in terms of power consumption and delay, are presented based on the optimal configuration of the wake-up parameters  (\ref{eq:A}). Namely, a) with simplified assumptions of zero false alarm/misdetection rates, and $t_{on}\approx 0$ ms, equivalent to  analytical results  (ana.), and b)   with the realistic assumptions of $P_{fa}=10 \%$, $P_{md}=1 \%$, $t_{on}=1/14$ ms obtained by simulations and \cite{Globecom}, referred to as simulation results (sim.).  

Table \ref{optimal_tw} shows the optimal resulting values of $t_w^*$ in (\ref{eq:A}) for different values of $\lambda$ and $\overline{\text{D}}_{\max}$. As it can be observed, for tight delay requirements ($\overline{\text{D}}_{\max}=30$ ms),  $t_w^*$ tends to be small, enabling the UE to reduce the duration of packet buffering. Interestingly, for mid range of packet arrival rates ($\lambda=0.1$ p/ms), optimal wake-up cycle for a given delay bound is shorter than for both lower and higher packet arrival rates. The justification is as follows.  For higher packet arrival rates,  $t_w^*$   becomes larger,  the reason being that the inactivity timer is ON most of the time. Therefore, the need for smaller wake-up cycles decreases and correspondingly higher energy overhead is induced. For lower packet arrival rates, in turn, the value of $t_w^*$ is higher due to the infrequent packet arrivals, hence achieving a smaller delay. 

\begin{table}[!t]
\scriptsize
\centering
\renewcommand{\arraystretch}{1.3}
\caption{{Optimal values of wake-up cycle under different delay requirements and packet arrival rates ($t_i^*=1$ ms)}}
\label{optimal_tw}
\begin{tabular}{|c|c|c|c|c|c|c|c|c|c|}
\hline
$\lambda$ [p/ms]& \multicolumn{3}{c|}{$0.01$} & \multicolumn{3}{c|}{$0.08$} & \multicolumn{3}{c|}{$0.15$} \\ \hline
$\overline{\text{D}}_{\max} $[ms] &$30$&$75$&$500$&$30$&$75$&$500$&$30$&$75$&$500$\\ \hline
$t_w^*$ [ms] & 180 & 380 &2099  & 124 & 315& 2124 & 125 &328 &2246  \\ \hline
\end{tabular}
\end{table}

{Fig.} \ref {fig:PC_lambda} illustrates the power consumption of the proposed WuS under ideal and realistic assumptions as a function of the packet arrival rate ($\lambda< \lambda_t$), for different maximum tolerable delays. 
As it can be observed, for both analytical and simulation results, and for all  delay bounds, the  average power consumption initially  increases,  while then remains almost constant (especially for large delay bounds) as $\lambda$ increases, due to the configuration of shorter  wake-up cycles for mid packet arrival rates  (see Table \ref{optimal_tw}). Moreover, the UE consumes higher power in order to satisfy tighter delay requirements, which, as shown in Table   \ref{optimal_tw}, can be translated into shorter $t_w^*$. Furthermore,  {Fig.} \ref {fig:PC_lambda} shows that the simulation results closely follow the analytical results, the non-zero gap being due to the non-zero false alarm and misdetection rates.   The relative gap between simulation-based and analytical results is somewhat larger for shorter delay bounds, which stems from the correspondingly higher number of wake-up instances.  

\begin{figure}[!t]
\centering
\includegraphics[width=0.39 \textwidth]{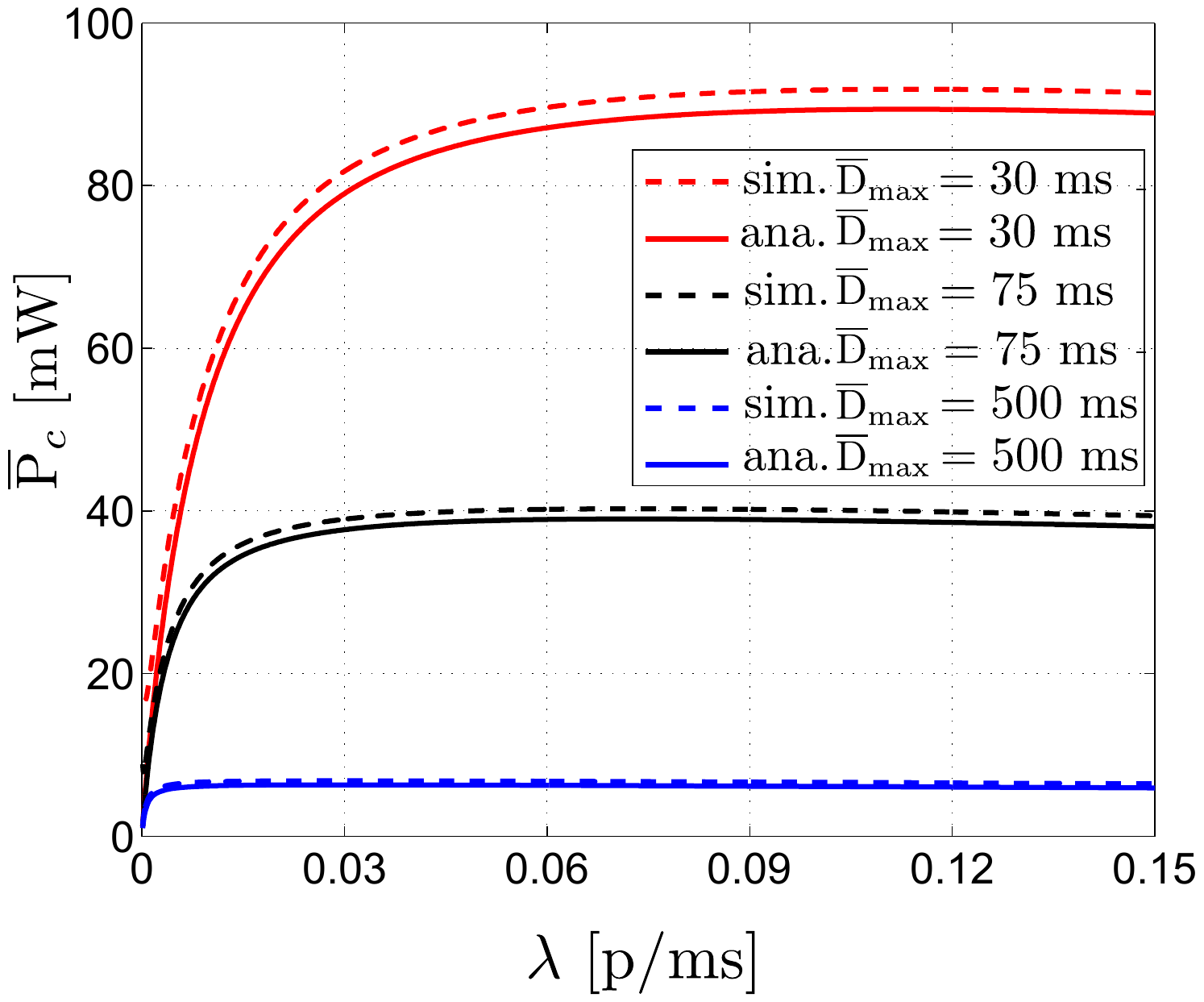}
\caption{{Average power consumption of the optimized wake-up scheme, under ideal and realistic assumptions as a function of packet arrival rate and delay bound.}}
\label{fig:PC_lambda}
\end{figure}
Moreover, {Fig.} \ref{fig:delay_lambda} depicts the average packet delay experienced by the WRx-enabled UE under ideal and realistic assumptions when  packet arrival rates vary. As it can be observed, the analytical  delay based on the optimal parameter configuration (\ref{eq:A}) is slightly shorter than the maximum tolerable delay. This is because of selecting  the greatest integer less than or equal to the optimal  wake-up cycle of the relaxed optimization problem. However, the actual average delay is slightly  higher than the analytical  average delay, especially for high delay bounds. The main reason for such negligible excess delay is the unavoidable misdetections, whose impact is more clear for large wake-up cycles corresponding to high delay bounds. In practice, to compensate for such small excess delay, the delay bound can be set slightly smaller than the actual average delay requirement.

\begin{figure}[!t]
\centering
\includegraphics[width=0.4 \textwidth]{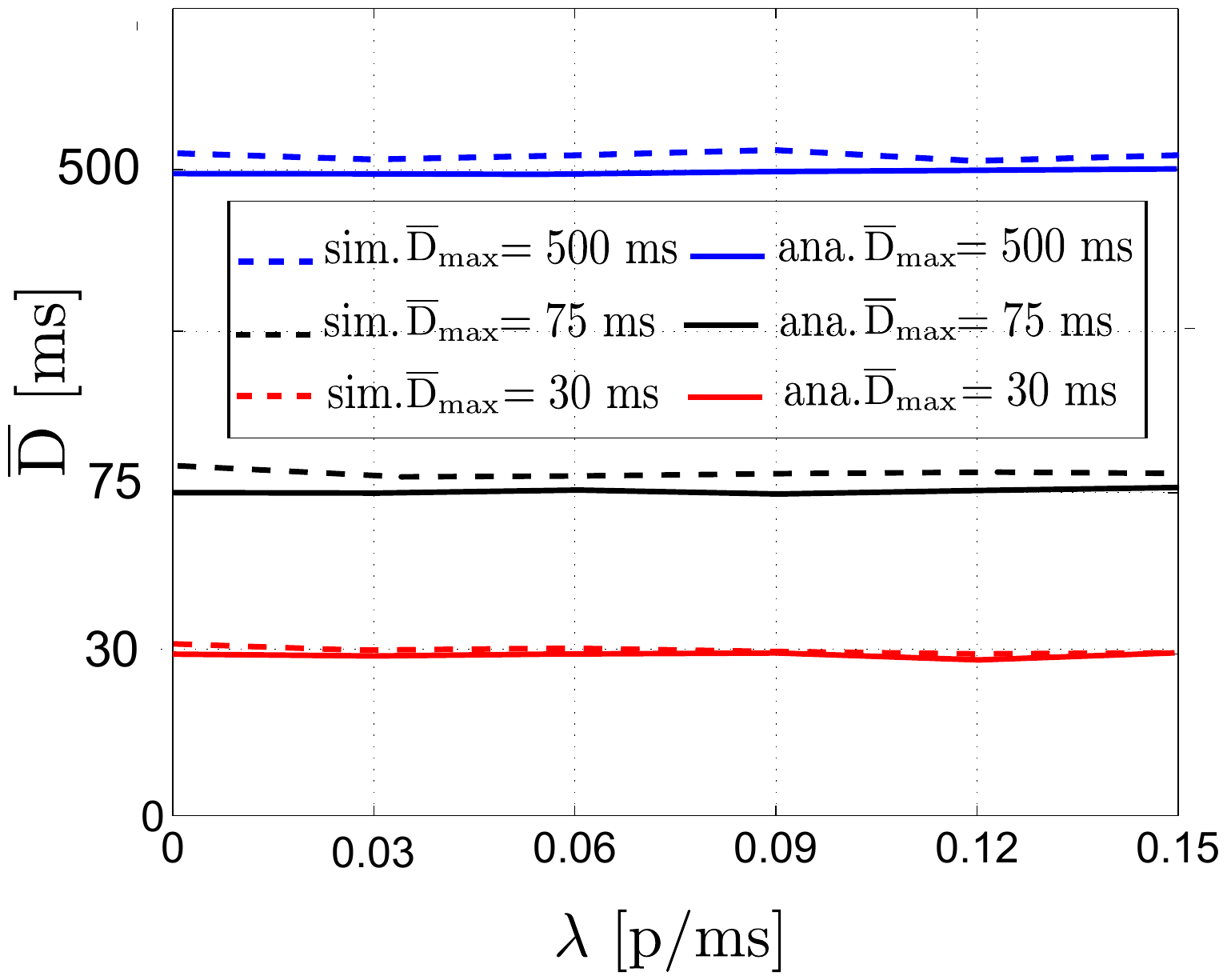}
\caption{Average buffering delay of the optimized wake-up scheme, under ideal and realistic assumptions as a function of packet arrival rate and delay bound. For better visualization, the y-axis is deliberately not in linear scale.}
\label{fig:delay_lambda}
\end{figure}

Finally, for comparison purposes, the relative power saving of WuS over DRX representing the amount of power that can be saved with WuS as compared to the DRX-based reference system is utilized, assuming the same delay constraints in both methods. The value of the relative power saving ranges from $0$ to $100\%$, and a large value
indicates that the WuS conserves energy better than the
DRX. Formally, we express the relative power saving ($\eta$) as 
  \begin{equation}\label{eq:etapgm}
\eta=\frac{\overline{\text{P}}_{\text{DRX}}-\overline{\text{P}}_c}{\overline{\text{P}}_{\text{DRX}}}\times 100,  
 \end{equation}
\noindent where  $\overline{\text{P}}_{\text{DRX}}$ refers to average power consumption of DRX. Furthermore, for a fair comparison, we consider an exhaustive search over a large parameter set of DRX configuration, developed by authors in \cite{Ramazanali}. However, in order to  take  start-up and power-down power consumption into account, the solution in \cite{Ramazanali} is slightly modified to account for the transitory states.  

\begin{figure}[!t]
\centering
\includegraphics[width=0.4 \textwidth]{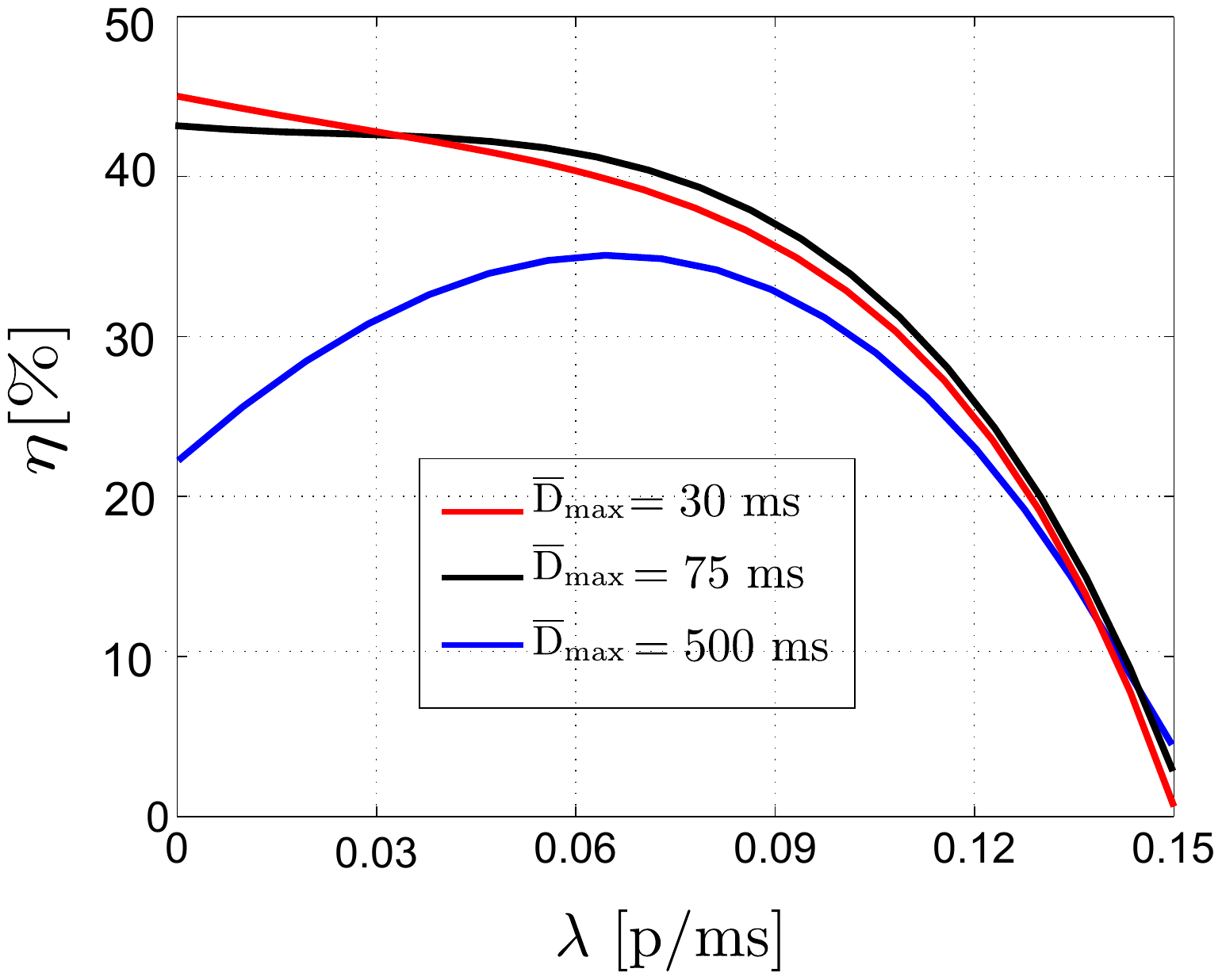}
\caption{{Achieved relative power saving values ($\eta$) of the proposed optimized wake-up scheme as a function  of packet arrival rate and delay bound.}}
\vspace{-3mm}
\label{fig:relative_pc}
\end{figure}

{Fig. \ref{fig:relative_pc}} shows the power saving results. It is observed that the proposed WuS, under realistic assumptions, outperforms DRX within the range $\lambda< \lambda_t$, especially for low packet arrival rates with tight delay requirements. 
The main reason is that in such scenarios, DRX-based device needs to decode the control channel very often, residing mainly in short DRX cycles, which causes extra power consumption. The WRx, in turn, needs to decode the wake-up signaling frequently, but with lower power overhead. Additionally, as expected, regardless of the delay requirements, for higher packet arrival rates, DRX infers relatively similar power consumption to the WuS. The reason is that in such cases, the DRX parameters can be configured in such a way that there is a low amount of unscheduled DRX cycles, either by utilizing  short DRX cycles for very tight delay bounds or by employing  long DRX cycles for large delay requirements. Overall, the results in Fig. \ref{fig:relative_pc} clearly demonstrate that WuS can provide substantial energy-efficiency improvements compared to DRX, with the maximum energy-savings being in the order of 40\%.   
\section{Discussions and Final Remarks}
\label{sec:remark}
In this section, three interesting remarks are drawn and discussed. 

\vspace{2mm}
\textbf{Remark 1:} The proposed WuS is fully independent of DRX, which means that both methods can co-exist, interact and be used together to reduce energy consumption of the UE even further. Based on the  numerical results provided in this paper, our opinion regarding the power saving mechanisms for    moderate and generic  mobile users is that    there is no 'One-Size-Fits-All Solution', unless the UE is well-defined and narrowed  to a specific application  and QoS requirement. For a broad range of applications and QoS requirements, there is a need for combining and utilizing different power saving mechanisms, and selecting the method that fits best for particular circumstances. For example, the WuS can be utilized for packet arrival rates lower than the turnoff packet arrival rate; for higher packet arrival rates or shorter delay bounds (e.g., smaller than $30$ ms), DRX may eventually be the preferable method of choice; for ultra-low latency requirements, other power saving mechanisms may be needed, building on, e.g., microsleep \cite{Lauridsenthesis} or pre-grant message \cite{pregrant,8616818} concepts. Further, depending on whether an RRC context is established or not, the WuS is agnostic to the RRC states, and can be adopted for idle (delay bound is in range of some hundreds of milliseconds), inactive, and connected modes.

\begin{table*} 
\small
\centering
\renewcommand{\arraystretch}{1.4}
\caption{{Minimum power consumption values [mW] as function of TTI size, delay bounds and packet arrival rates} }
\label{tti}
\begin{tabular}{|c|c|c|c|c|c|c|c|c|c|}
\hline
$\lambda$ [p/ms]& \multicolumn{3}{c|}{$0.01$} & \multicolumn{3}{c|}{$0.08$} & \multicolumn{3}{c|}{$0.15$} \\ \hline
$\overline{\text{D}}_{\max} $[ms] &$30$&$75$&$500$&$30$&$75$&$500$&$30$&$75$&$500$\\ \hline
$\overline{\text{P}}_c$ @ TTI$~=1$ ms& 54.2 &31.6&6.6 &88.7 & 39& 6.1 &88.9&38.1&5.9  \\ \hline
$\overline{\text{P}}_c$ @ TTI$~=500$ $\mu$s&  50.4&29.4 &5.7 & 83.7  &36.8& 5.8 &85.4 &36.6  &5.7\\ \hline
$\overline{\text{P}}_c$ @ TTI$~=250$ $\mu$s&  48.3& 28.1 &5.5& 81.1 &35.7& 5.7& 83.7&35.9 &5.6 \\ \hline
$\overline{\text{P}}_c$ @ TTI$~=125$ $\mu$s&  47.5&27.7 &5.4& 80.1 &35.3& 5.6& 83.1&35.7 &5.6 \\ \hline
\end{tabular}
\end{table*}

\vspace{2mm}
\textbf{Remark 2:} As mentioned in Section  \ref{sec:sysmodel}, a  latency-optimized frame structure  with flexible numerology    is adopted in 5G NR, for which the  {slot length scales down when the numerology increases} \cite{nr2017}.  In this work,    different time intervals within the WuS  are defined as multiples of a   time unit of a TTI with a duration of $1$ ms. As shown before, the minimum power consumption over the boundary is limited by the minimum feasible value of $t_i$. Therefore, if the TTI can be selected even smaller, i.e., with finer granularity, the optimal power consumption can be further reduced. In this line, Table \ref{tti} presents the power consumption with different TTI sizes (corresponding to NR numerologies 0, 1, 2, and 3~\cite{TS38300}), delay bounds, and packet arrival rates.

Interestingly, Table \ref{tti} shows how the  {5G NR numerologies} facilitate the use of WuS and improve the
applicability and energy saving potential of WuS compared to longer TTI sizes. Besides the smaller $t_i$ sizes, with smaller TTIs, the corresponding optimal wake-up cycles are more fine-grained.  With shorter TTI sizes down to $125$ $\mu$s,  the proposed  WuS can  provide up to  $12 \% $ additional energy savings compared to the baseline $1$ ms TTI.  Therefore, on average, the benefit of the flexible NR frame structure is not only for low latency communication but it can also offer energy savings depending on the traffic arrival rates and delay constraints.

\vspace{2mm}
\textbf{Remark 3:} The traffic model assumed in this article is basically well-suited for  the periodic nature of DRX. For instance,  voice calls and video streaming have such periodic behaviour. However, in other application areas such as  MTCs, where sensors can be aperiodically polled by either a user or a machine, the traffic will have more non-periodic patterns. In such case, the DRX may not fit well, while the WuS has more suitable characteristics, being more robust and agnostic  to the traffic type.  
\section{Conclusions and Future Work}
\label{sec:conc}
In this article, wake-up based downlink access under delay constraints was studied in the context of 5G NR networks, with particular focus on energy-efficiency optimization. It was shown that the performance of the wake-up scheme   is governed by a set of two parameters that interact with each other in an intricate manner. To find the optimal wake-up parameters configuration, and thus to take full advantage of the power saving capabilities of the wake-up scheme, a constrained optimization problem was formulated, together with the corresponding closed-form solution. Analytical and simulation  results  showed that the proposed scheme is an efficient approach to  reduce  the device energy consumption, while ensuring a predictable and consistent latency. The numerical results also showed that the optimized wake-up system outperforms the corresponding optimized DRX-based reference system in power efficiency. Furthermore,  the range of packet arrival rates within which the WuS works efficiently was established, while outside that range other power saving mechanisms, such as DRX or microsleep, can be used.

Future work includes extending the proposed framework to bidirectional communication scenarios with the corresponding downlink and uplink traffic patterns and the associated QoS requirements, as well as to consider other realistic assumptions that impact the energy-delay trade-offs, such as the communication rate and scheduling delays. Additionally, an interesting aspect is to investigate how to configure the wake-up scheme parameters for application-specific traffic scenarios, such as virtual and augmented reality, when both uplink and downlink traffics are considered.  Finally, focus can also be given to optimizing the wake-up scheme parameters based on the proposed framework, by utilizing not only traffic statistics but also short-term traffic pattern prediction by means of modern  machine learning methods.  
\section*{Appendix: Analysis of $\overline{\text{P}}_{b}(t_w)$}
\addcontentsline{toc}{section}{Appendices}
\renewcommand{\thesubsection}{\Alph{subsection}}

In order to find the optimal value of $t_w$ and correspondingly $t_i$, the derivation of $\overline{\text{P}}_{b}(t_w)$ with respect to $t_w$ is given as follows 
 {  \begin{equation}
\frac{d \overline{\text{P}}_{b}(t_w)}{d t_w}=\text{PW}_3\frac{Y(t_w)}{\big ({w_1+w_2t_w+w_3e^{-\lambda t_w}}\big )^2},
 \label{eq:Awe}
\end{equation}
\noindent where
  \begin{align}
   \label{eq:uuu}
Y(t_w)& =F_1+(F_2-\lambda F_3t_w)e^{-\lambda t_w} ,\\
\label{eq:f1}
  F_1&=(1+e^{\lambda t_s})(\overline{\text{D}}_{\max}-\frac{t_s}{2})(t_s e^{\lambda t_s}+\frac{{1}}{\lambda}) \times \nonumber \\
& \quad  \big(\frac{1}{2}(t_{pd}-\phi t_{su})+\frac{{1}}{\lambda}+2(\overline{\text{D}}_{\max}-\frac{t_s}{2})\big)-\nonumber\\
& \quad \big (\frac{1}{2}(\phi t_{su}+t_{pd})-\frac{{1}}{\lambda}\big )(1+e^{\lambda t_s})^2(\overline{\text{D}}_{\max}-\frac{t_s}{2})^2,  \\
      \label{eq:f2}
F_2&=-F_1+(1+e^{\lambda t_s})(\overline{\text{D}}_{\max}-\frac{t_s}{2})^2 \times \nonumber \\ & \quad (t_s e^{\lambda t_s}+\frac{{1}}{{\lambda}})\big(-2+\lambda(2-\phi) t_{su}+\lambda t_{pd}\big), \\
 \label{eq:f3}
F_3&= F_1-2(1+e^{\lambda t_s})(\overline{\text{D}}_{\max}-\frac{t_s}{2})^2(t_s e^{\lambda t_s}+\frac{{1}}{{\lambda}}).
\end{align}}
{Based on typical values of $1\leq\phi< 2$, the condition $0 \leq (2-\phi) t_{su}+t_{pd}$ is met, and hence based on (\ref{eq:f1})-(\ref{eq:f3}), we can conclude that $F_1+F_2> 0, $  $F_2+F_3> 0 $ and  $F_1 > F_3$.  Furthermore, it can be shown that $F_1$ is a decreasing function of $\lambda$, which has a single root. We refer to its root as $\lambda_t$; where for all $\lambda <\lambda_t$, then $F_1>0$ while, for all $\lambda >\lambda_t$, then $F_1<0$. Root-finding algorithms can be utilized to find $\lambda_t$ as the $\lambda$ value that meets $F_1=0$.}
 
{To  determine whether $\frac{d \overline{\text{P}}_{b}(t_w)}{d t_w}$ is positive or negative, $Y(t_w)$ needs to be analyzed. 
\noindent By differentiating $Y(t_w)$ with respect to $t_w$, we  obtain 
  \begin{align}
\frac{{d} Y(t_w)}{d t_w}=-\lambda(F_2+F_3-\lambda F_3 t_w)e^{-\lambda t_w}.
 \label{eq:Awer}
\end{align}}
{Additionally, depending on whether $F_3$ and $F_1$ are positive or negative, the $\overline{\text{P}}_{b}(t_w)$ behaves differently. In order to characterize the behaviour of $\overline{\text{P}}_{b}(t_w)$, we define three mutually exclusive cases: Case A ($F_3>0$), Case B ($F_3<0$ and $F_1>0$), and Case C ($F_3<0$ and $F_1<0$). Due to fact that  $F_1 > F_3$,  in the former case,  $F_1$ is always positive.}

Case A ($F_3>0$): {Based on (\ref{eq:Awer}),  if $F_3$ is positive, $Y(t_w)$ is a decreasing function for the range of $t_w<\frac{F_2+F_3}{\lambda F_3}$ and an increasing function for $t_w>\frac{F_2+F_3}{\lambda F_3}$. Therefore, $Y(t_w)$ has a minimum point at $t_w=\frac{F_2+F_3}{\lambda F_3}$, where the $Y(t_w)$  at $t_w=\frac{F_2+F_3}{\lambda F_3}$ is positive, i.e.,   $Y(\frac{F_2+F_3}{\lambda F_3})=F_1-F_3e^{-\lambda t_w}>0$  (due to fact that $F_3<F_1$). As a result, $Y(t_w)$ is always positive and hence $\overline{\text{P}}_{b}(t_w)$ is a monotonous increasing function.}

{Based on (\ref{eq:Awer}),  if $F_3$ is negative (Case B or Case C), we can conclude that $Y(t_w)$ is a monotonous  decreasing function from $F_1+F_2>0$ to $F_1$.} 
  
Case B ($F_3<0$ and $F_1>0$):  {In this case, for all values of wake-up cycle, $Y(t_w)$ is positive, and hence $\overline{\text{P}}_{b}$  is a  monotonous increasing function.}
  
 Case C ($F_3<0$ and $F_1<0$): {In this case, $Y(t_w)$ for $t_w<t_{ws}$ is positive,  and it is  negative for $t_w>t_{ws}$; where $t_{ws}$ is a stationary point, i.e., $Y(t_{ws})=0$ or equivalently, $\frac{\partial \overline{\text{P}}_{b}(t_w)}{\partial t_w}\vert _{t_w=t_{w_s}}=0$. As a result, $\overline{\text{P}}_{b}(t_w)$ is an increasing function within $t_w<t_{ws}$, and a decreasing function for $t_w>t_{ws}$. For the typical range of parameters, we have consistently observed through simulations that $t_{w_s}<t_{w_b}$, so that   we can conclude that   $\overline{\text{P}}_{b}(t_w)$ is a decreasing function for the feasible range of the wake-up cycle (i.e., $t_{w}>t_{w_b}$).}

{To sum up,  for $F_1<0$, or equivalently,   ${\lambda}_t<{\lambda}$ (Case C), $\overline{\text{P}}_{b}(t_w)$ is a monotonous decreasing function  while, for ${\lambda}<{\lambda}_t$ (Case A or B), $\overline{\text{P}}_{b}(t_w)$ is a monotonous increasing function. Due to the relevance of ${\lambda}_t$, we refer to it as the turnoff packet arrival rate. }

\section*{Acknowledgment}
This work has received funding from the European Union's Horizon 2020 research and innovation program under the Marie Sk\l{}odowska-Curie grant agreement No. 675891 (SCAVENGE), Tekes TAKE-5 project,  {Spanish MINECO grant
TEC2017-88373-R (5G-REFINE), and Generalitat de Catalunya grant 2017 SGR 1195.}

\bibliography{Master}

\begin{thebibliography}{10}

\bibitem{Globecomnew}
S.~Rostami, S.~Lagen, M.~Costa, P.~Dini, and M.~Valkama, ``Optimized wake-up
  scheme with bounded delay for energy-efficient {MTC},'' in {\em IEEE Global
  Communications Conference (GLOBECOM)}, pp.~1--6, Dec. 2019.

\bibitem{nr2017}
S.~Parkvall, E.~Dahlman, A.~Furuskar, and M.~Frenne, ``N{R}: The new 5{G} radio
  access technology,'' {\em IEEE Communications Standards Magazine}, vol.~1,
  pp.~24--30, Dec 2017.

\bibitem{Boccardi2014}
F.~Boccardi, R.~W. Heath, A.~Lozano, T.~L. Marzetta, and P.~Popovski, ``Five
  disruptive technology directions for 5{G},'' {\em IEEE Communications
  Magazine}, vol.~52, pp.~74--80, February 2014.

\bibitem{Fehske}
A.~{Fehske}, G.~{Fettweis}, J.~{Malmodin}, and G.~{Biczok}, ``The global
  footprint of mobile communications: The ecological and economic
  perspective,'' {\em IEEE Communications Magazine}, vol.~49, pp.~55--62,
  August 2011.

\bibitem{Qualcomm2013}
``Designing mobile devices for low power and thermal efficiency,'' tech. rep.,
  Qualcomm Technologies, Inc., Oct. 2013.

\bibitem{Lauridsenthesis}
M.~Lauridsen, {\em Studies on Mobile Terminal Energy Consumption for {LTE} and
  Future {5G}}.
\newblock PhD thesis, Aalborg University, Jan. 2015.

\bibitem{Carroll2010}
A.~Carroll and G.~Heiser, ``An analysis of power consumption in a smartphone,''
  in {\em Proc. USENIXATC2010}, (Berkeley, CA, USA), pp.~21--21, USENIX
  Association, 2010.

\bibitem{Lauridsen2015}
M.~Lauridsen, P.~Mogensen, and T.~B. Sorensen, ``Estimation of a 10 {G}b/s 5{G}
  receiver's performance and power evolution towards 2030,'' in {\em 2015 IEEE
  82nd Vehicular Technology Conference (VTC2015-Fall)}, pp.~1--5, Sept 2015.

\bibitem{ITU2015}
``{IMT} traffic estimates for the years 2020 to 2030,'' tech. rep., ITU-R, July
  2015.

\bibitem{TS38300}
``{TS} 38.300, 3rd {G}eneration {P}artnership {P}roject; {T}echnical
  {S}pecification {G}roup {R}adio {A}ccess {N}etwork; {NR}; {NR} and {NG-RAN}
  overall description,'' tech. rep., 3GPP, Jan. 2019.

\bibitem{erik}
S.~P. Erik~Dahlman and J.~Skold, {\em 4G LTE/LTE-Advanced for Mobile
  Broadband}.
\newblock Academic Press, 2011.

\bibitem{tr36.213}
``{LTE}; {E}volved {U}niversal {T}errestrial {R}adio {A}ccess ({E-UTRA});
  {P}hysical layer procedures,'' tech. rep., 3GPP TS 36.213 version 10.1.0
  Release 10, Apr 2010.

\bibitem{TS38.30}
``5{G};{NR}; {U}ser {E}quipment ({UE}) procedures in idle mode and in {RRC}
  inactive state,'' tech. rep., 3GPP TS 38.304 version 15.1.0 Release 15, Oct.
  2018.

\bibitem{Liu}
E.~Liu, J.~Zhang, and W.~Ren, ``Adaptive and autonomous power-saving scheme for
  beyond 3{G} user equipment,'' {\em IET Communications}, vol.~7, pp.~602--610,
  May 2013.

\bibitem{koc}
A.~T. {Koc}, S.~C. {Jha}, R.~{Vannithamby}, and M.~{Torlak}, ``Device power
  saving and latency optimization in {LTE-A} networks through {DRX}
  configuration,'' {\em IEEE Transactions on Wireless Communications}, vol.~13,
  pp.~2614--2625, May 2014.

\bibitem{Mihov}
Y.~Y. Mihov, K.~M. Kassev, and B.~P. Tsankov, ``Analysis and performance
  evaluation of the {DRX} mechanism for power saving in {LTE},'' in {\em 2010
  IEEE 26-th Convention of Electrical and Electronics Engineers in Israel},
  pp.~000520--000524, Nov 2010.

\bibitem{Ramazanali}
H.~Ramazanali and A.~Vinel, ``Tuning of {LTE/LTE-A DRX} parameters,'' in {\em
  Proc. IEEE CAMAD 2016}, pp.~95--100, Oct 2016.

\bibitem{Demirkol2009}
I.~Demirkol, C.~Ersoy, and E.~Onur, ``Wake-up receivers for wireless sensor
  networks: benefits and challenges,'' {\em IEEE Wireless Communications},
  vol.~16, pp.~88--96, Aug 2009.

\bibitem{Mazloum2014}
N.~S. Mazloum and O.~Edfors, ``Performance analysis and energy optimization of
  wake-up receiver schemes for wireless low-power applications,'' {\em IEEE
  Transactions on Wireless Communications}, vol.~13, pp.~7050--7061, Dec 2014.

\bibitem{Globecom}
S.~Rostami, K.~Heiska, O.~Puchko, J.~Talvitie, K.~Leppanen, and M.~Valkama,
  ``Novel wake-up signaling for enhanced energy-efficiency of 5{G} and beyond
  mobile devices,'' in {\em 2018 IEEE Global Communications Conference
  (GLOBECOM)}, pp.~1--7, Dec 2018.

\bibitem{Lauridsen2016}
M.~Lauridsen, G.~Berardinelli, F.~M.~L. Tavares, F.~Frederiksen, and
  P.~Mogensen, ``Sleep modes for enhanced battery life of 5{G} mobile
  terminals,'' in {\em Proc. IEEE VTC 2016 Spring}, pp.~1--6, May 2016.

\bibitem{TS36.300}
``Evolved universal terrestrial radio access ({E-UTRA}) and evolved universal
  terrestrial radio access network ({E-UTRAN}); overall description,'' tech.
  rep., 3GPP, TS 36.300, Mar. 2019.

\bibitem{NR_PS}
``New {SID}: Study on {UE} power saving in {NR},'' tech. rep., 3GPP, RP-181463,
  June. 2018.

\bibitem{wpwrx}
S.~{Rostami}, K.~{Heiska}, O.~{Puchko}, K.~{Leppanen}, and M.~{Valkama},
  ``Wireless powered wake-up receiver for ultra-low-power devices,'' in {\em
  2018 IEEE Wireless Communications and Networking Conference (WCNC)},
  pp.~1--5, April 2018.

\bibitem{Tang}
S.~{Tang}, H.~{Yomo}, Y.~{Kondo}, and S.~{Obana}, ``Wakeup receiver for
  radio-on-demand wireless lans,'' in {\em 2011 IEEE Global Telecommunications
  Conference (GLOBECOM)}, pp.~1--6, Dec 2011.

\bibitem{Wilhelmsson}
L.~R. {Wilhelmsson}, M.~M. {Lopez}, S.~{Mattisson}, and T.~{Nilsson},
  ``Spectrum efficient support of wake-up receivers by using {(O)FDMA},'' in
  {\em 2018 IEEE Wireless Communications and Networking Conference (WCNC)},
  pp.~1--6, April 2018.

\bibitem{Kouzayha}
N.~{Kouzayha}, Z.~{Dawy}, J.~G. {Andrews}, and H.~{ElSawy}, ``Joint
  downlink/uplink {RF} wake-up solution for {I}o{T} over cellular networks,''
  {\em IEEE Transactions on Wireless Communications}, vol.~17, pp.~1574--1588,
  March 2018.

\bibitem{Oller}
J.~{Oller}, E.~{Garcia}, E.~{Lopez}, I.~{Demirkol}, J.~{Casademont},
  J.~{Paradells}, U.~{Gamm}, and L.~{Reindl}, ``{IEEE} 802.11-enabled wake-up
  radio system: design and performance evaluation,'' {\em Electronics Letters},
  vol.~50, pp.~1484--1486, Sep. 2014.

\bibitem{Aoudia}
F.~{Ait Aoudia}, M.~{Magno}, M.~{Gautier}, O.~{Berder}, and L.~{Benini},
  ``Analytical and experimental evaluation of wake-up receivers based
  protocols,'' in {\em 2016 IEEE Global Communications Conference (GLOBECOM)},
  pp.~1--7, Dec 2016.

\bibitem{ponna2018saving}
R.~P.~K. Ponna and D.~Ray, ``Saving energy in cellular {IoT} using low-power
  wake-up radios,'' 2018.
\newblock Student Paper, Lund University.

\bibitem{eee1}
``Low-power wake-up receiver for 802.11,'' tech. rep., IEEE 802.11-15/1307r1.

\bibitem{eee5}
``Proposal for wake-up receiver ({WUR}) study group,'' tech. rep., IEEE
  802.11-16/0722r0.

\bibitem{pm2}
C.~C. Tseng, H.~C. Wang, F.~C. Kuo, K.~C. Ting, H.~H. Chen, and G.~Y. Chen,
  ``Delay and power consumption in {LTE/LTE-A DRX} mechanism with mixed short
  and long cycles,'' {\em IEEE Transactions on Vehicular Technology}, vol.~65,
  pp.~1721--1734, March 2016.

\bibitem{mnlp}
P.~Belotti, C.~Kirches, S.~Leyffer, J.~Linderoth, J.~Luedtke, and A.~Mahajan,
  ``Mixed-integer nonlinear optimization,'' {\em Acta Numerica}, vol.~22,
  p.~1–131, 2013.

\bibitem{Corless1996}
R.~M. Corless, G.~H. Gonnet, D.~E.~G. Hare, D.~J. Jeffrey, and D.~E. Knuth,
  ``On the {L}ambertw function,'' {\em Advances in Computational Mathematics},
  vol.~5, pp.~329--359, Dec 1996.

\bibitem{pm1}
T.~{Kolding}, J.~{Wigard}, and L.~{Dalsgaard}, ``Balancing power saving and
  single user experience with discontinuous reception in lte,'' in {\em 2008
  IEEE International Symposium on Wireless Communication Systems},
  pp.~713--717, Oct 2008.

\bibitem{Lauridsen2014}
M.~Lauridsen, L.~Noel, T.~Sorensen, and P.~Mogensen, ``An empirical {LTE}
  smartphone power model with a view to energy efficiency evolution,'' {\em
  Intel Technology Journal}, vol.~18, pp.~172--193, 3 2014.

\bibitem{pregrant}
S.~{Rostami}, K.~{Heiska}, O.~{Puchko}, K.~{Leppanen}, and M.~{Valkama},
  ``Robust pre-grant signaling for energy-efficient {5G} and beyond mobile
  devices,'' in {\em 2018 IEEE International Conference on Communications
  (ICC)}, pp.~1--6, May 2018.

\bibitem{8616818}
S.~{Rostami}, K.~{Heiska}, O.~{Puchko}, K.~{Leppanen}, and M.~{Valkama},
  ``Pre-grant signaling for energy-efficient 5{G} and beyond mobile devices:
  method and analysis,'' {\em IEEE Transactions on Green Communications and
  Networking}, vol.~3, pp.~418--432, June 2019.

\end{thebibliography}
\bibliographystyle{ieeetr}

\end{document}